\documentclass[10pt,journal,compsoc]{IEEEtran}
\usepackage{times}
\usepackage{amsmath}
\usepackage{amsthm}
\usepackage{amssymb}
\usepackage{mathtools}
\usepackage{epsfig}
\usepackage{graphicx}
\usepackage{algorithm}
\usepackage{algpseudocode}
\usepackage{balance}
\usepackage[colorlinks,
            linkcolor=black,
            anchorcolor=black,
            citecolor=black
            ]{hyperref}
\newenvironment{sequation}{\begin{equation}\footnotesize}{\end{equation}}
\newenvironment{salgorithm}{\begin{algorithm}[!t]\footnotesize}{\end{algorithm}}
\allowdisplaybreaks

\newtheorem{definition}{Definition}
\newtheorem{theorem}{Theorem}[section]

\newtheorem{pruning rule}[theorem]{Pruning Rule}
\newcommand{\nop}[1]{}

\ifCLASSOPTIONcompsoc
  \usepackage[nocompress]{cite}
\else
  \usepackage{cite}
\fi
\ifCLASSINFOpdf
\else
\fi
\begin{document}
%
\title{Reducing Uncertainty of Schema Matching via Crowdsourcing with Accuracy Rates}
%
%
%
%

\author{Chen Jason Zhang,
        Lei Chen,~\IEEEmembership{Member,~IEEE,}
        H. V. Jagadish,~\IEEEmembership{Member,~IEEE,}
        Mengchen Zhang,
        and~Yongxin~Tong,~\IEEEmembership{Member,~IEEE,}
\IEEEcompsocitemizethanks{\IEEEcompsocthanksitem Chen Jason Zhang is with School of Computer Science and Technology, Shandong University of Finance and Economics, Jinan, Shandong, China and Department of Computer
Science and Engineering, the Hong Kong University of Science and
Technology, Kowloon, Hong Kong, SAR China\protect\\
E-mail: czhangad@cse.ust.hk
\IEEEcompsocthanksitem Lei Chen and Mengchen Zhang are with Department of Computer
Science and Engineering, the Hong Kong University of Science and Technology, Kowloon, Hong Kong, SAR China\protect\\
E-mail: leichen@cse.ust.hk, mzhangag@connect.ust.hk
\IEEEcompsocthanksitem H. V. Jagadish is with  Department of Electrical Engineering and
Computer Science, University of Michigan, USA\protect\\
E-mail: jag@eecs.umich.edu
\IEEEcompsocthanksitem Yongxin Tong is with State Key Laboratory of Software Development Environment, School of Computer Science and Engineering, Beihang University, Beijing, China.\protect\\
E-mail: yxtong@buaa.edu.cn}
\thanks{Manuscript received April 19, 2005; revised August 26, 2015.}}

%
%

\markboth{Journal of \LaTeX\ Class Files,~Vol.~14, No.~8, August~2015}%
{Shell \MakeLowercase{\textit{et al.}}: Bare Demo of IEEEtran.cls for Computer Society Journals}
%



\IEEEtitleabstractindextext{%
\begin{abstract}
Schema matching is a central challenge for data integration systems.
Inspired by the popularity and the success of crowdsourcing platforms, we explore the use of
crowdsourcing to reduce the uncertainty of schema matching. Since crowdsourcing platforms are most effective for simple questions, we
assume that each \textit{Correspondence Correctness
Question} (CCQ) asks the crowd to decide whether a given
correspondence should exist in the correct matching. Furthermore, members of a crowd may sometimes return incorrect answers with different probabilities. Accuracy rates of individual crowd workers are probabilities of returning correct answers which can be attributes of CCQs as well as evaluations of individual workers. We prove that uncertainty reduction equals to entropy of answers minus entropy of crowds and show how to obtain lower and upper bounds for it. We propose frameworks and efficient algorithms to dynamically manage the CCQs to maximize the uncertainty reduction within a limited budget of questions. We develop two novel approaches, namely ``Single CCQ'' and ``Multiple CCQ'', which
\textit{adaptively} select, publish and manage questions.  We verify
the value of our solutions with simulation and real implementation.
\end{abstract}

}

\maketitle

\IEEEdisplaynontitleabstractindextext

%
\IEEEpeerreviewmaketitle

\vspace{-1.5ex}
\IEEEraisesectionheading{\section{Introduction}}
\label{Into}
\vspace{-1ex}
\subsection{Background and Motivation}
\vspace{-1ex}
Schema matching refers to finding correspondences between elements
of two given schemata, which is a critical issue for many
database applications such as data integration, data warehousing,
and electronic commerce \cite{Rahm:2001:SAA:767149.767154}.
Figure~\ref{fig:schemas} illustrates a running example of the schema
matching problem: given two relational schemata $A$ and $B$ describing faculty
information, we aim to determine the correspondences (indicated by
dotted lines), which identify attributes representing the same
concepts in the two. There has been significant
work in developing automated algorithms for schema matching (please
refer to \cite{Rahm:2001:SAA:767149.767154} \cite{shvaiko2005survey} \cite{Bersteinsurvey2011} \cite{schemabook} for  comprehensive
surveys). Most approaches use linguistic, structural and
instance-based information. In general, it is still very difficult
to tackle schema matching completely with an algorithmic approach:
some ambiguity remains. This ambiguity is unlikely to be removed because it is
believed that typically ``the syntactic representation of schemata and data do not completely
convey the semantics of different databases''
\cite{DBLP:conf/vldb/MillerHH00}.

Given this inherent ambiguity, many schema matching tools will produce not just one matching, but rather a whole set of possible matchings.  In fact, there is even a stream of work dealing with
models of possible matchings, beginning with \cite{Dong:2007:DIU:1325851.1325930}.  The matching tool can produce a result similar to the upper part of Table
\ref{table:schemA}, with one matching per row, associated with a probability that it is the correct matching.
\begin{figure}[!t]
   \centerline{\psfig{figure=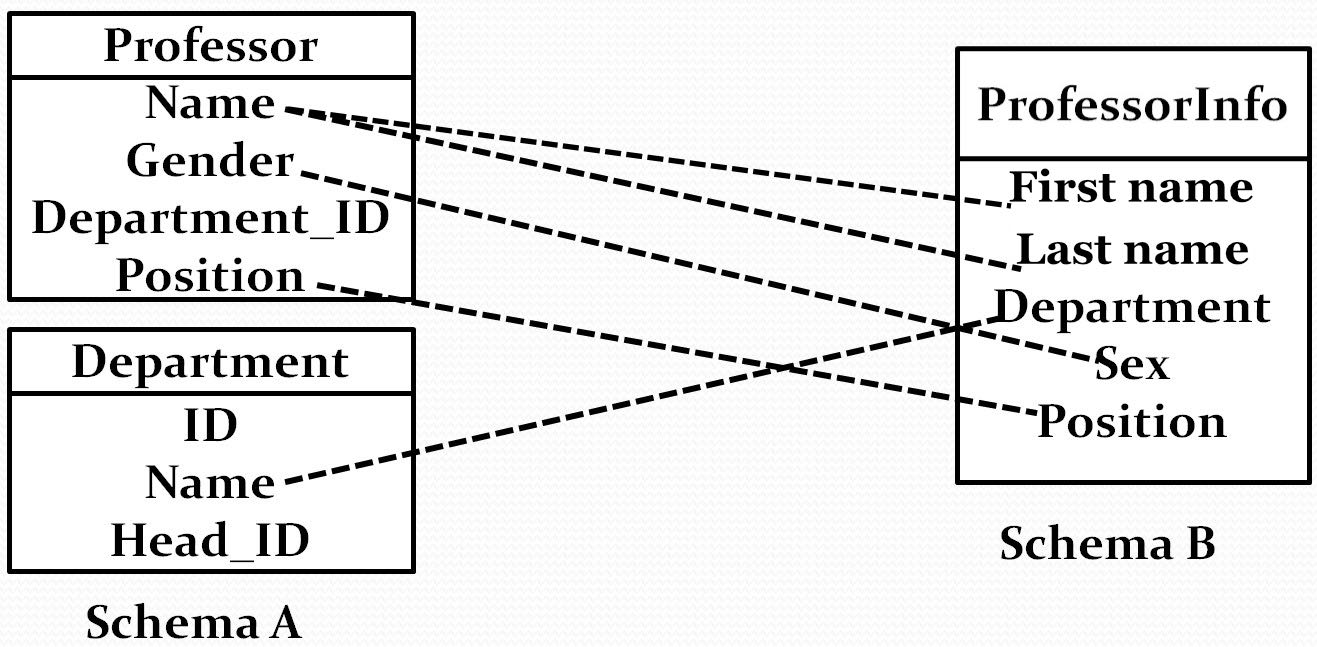,height=22mm} }
         \vspace{-3ex}
   \caption{\footnotesize{Example of Schema Matching Problem}}
   \label{fig:schemas}
      \vspace{-3ex}
\end{figure}

\begin{table}[!t]\caption{\footnotesize{Uncertain Schema Matching}}
\vspace{-3ex}
\label{two given schemas}
\footnotesize
\tabcolsep=0.01cm
\begin{tabular}{c c} 
\hline 
Possible Matchings & probability \\ 
$m_1$=\{ \textless(Professor)Name,[first name, last name] \textgreater,\\
\textless Position, Position\textgreater, \textless Gender,Sex\textgreater,& .45\\
\textless (Department) Name, Department\textgreater \} \\ 
$m_2$=\{ \textless(Professor)Name,[first name, last name] \textgreater,\\
\textless Gender, Sex\textgreater, \textless(Department) Name, Department\textgreater \}&.3 \\
$m_3$=\{ \textless(Department)Name, first name\textgreater, \textless Position, Position\textgreater\\
\textless Gender,Sex \textgreater \}&.25 \\
\end{tabular}\label{table:schemA} 
\begin{tabular}{c c} 
\hline 
Correspondence & probability \\ 
$c_1$=\textless(Professor)Name,[first name, last name] \textgreater&.75 \\
$c_2$=\textless Position, Position\textgreater&.7\\
$c_3$=\textless Gender,Sex \textgreater& 1 \\
$c_4$=\textless (Department) Name, Department\textgreater &.75\\
$c_5$=\textless (Department)Name,first name\textgreater &.25\\
\end{tabular}\label{table:schemB} 
\vspace{-6.1ex}
\end{table}

Given a set of possible matchings, one can create an integrated database that has uncertain data, and work with this using any of several systems that support \textit{probabilistic query processing} over uncertain data, such as \cite{DBLP:conf/sigmod/HuangAKO09}\cite{4812509}.
However, preserving the uncertainty complicates query processing and increases storage
cost.  So we would prefer to make choices earlier, if possible, and eliminate (or reduce) the uncertainty to be propagated.  It has been suggested \cite{Popa02translatingweb} that
\textbf{human insights are extremely conducive for reducing the uncertainty of schema matching},
so the correct matching can be manually chosen by the user from among the possible matchings offered by the system. In a traditional back-end database environment, where the human `user' is a DBA, setting up a new integrated database, such a system can work well.

However, in today's world, with end-users performing increasingly sophisticated data accesses, we have to support users who are interested, say, in combining data from two different web sources, and hence require an `ad hoc' schema matching.  Such users may not be experts, and will typically have little knowledge of either source schema. They may not even know what a schema is. They are also likely to have little patience with a system that asks them to make difficult choices, rather than just giving them the desired answer.  In other words, users may not themselves be a suitable source of human insight to resolve uncertainty in schema matching.

Fortunately, we have crowdsourcing technology as a promising option today. Many recent works, such as \cite{hung2013leveraging}, \cite{hung2014reconciling}, \cite{fan2014hybrid} and \cite{nguyen2014pay}, have suggested leveraging the crowd to improve schema matching. Platforms such as Amazon Mechanical Turk provide convenient access to crowds. The data concerning an explicit problem can be queried
by publishing questions, named Human Intelligent Tasks (a.k.a HITs).
The work-flow of publishing HITs can be automated with available
APIs (e.g. REST APIs) \cite{DBLP:conf/sigmod/FranklinKKRX11}.
To the extent that our end-user is not an expert, the opinion of a crowd of other non-experts is likely to be better than that of our end-user.
\vspace{-3ex}
\subsection{Problem Formulation and Contributions}
\vspace{-1ex}
It is well-known that crowdsourcing works best when tasks can be broken down into very simple pieces.  An entire schema matching task may be too large a grain for a crowd -- each individual may have small quibbles with a proposed matching, so that a simple binary question on the correctness of matchings may get mostly negative answers, with each user declaring it less than perfect.  On the other hand, asking open-ended questions is not recommended for a crowd, because it may be difficult to pull together a schema matching from multiple suggestions.   We address this challenge by posing to the crowd questions regarding individual correspondences for pairs of attributes, one from each schema being matched.   This much simpler question, in most circumstances, can be answered with a simple yes or no.
Of course, this requires that we build the machinery to translate between individual attribute correspondences and possible matchings.  Fortunately, this has been done before, in \cite{Dong:2007:DIU:1325851.1325930}, and is quite simple: since schema match options are all mutually exclusive, we can determine the probability of each correspondence by simply adding up the probabilities of matchings in which the correspondence holds.

Our problem then is to choose wisely the correspondences to ask the crowd to obtain the highest certainty of correct schema matching at the lowest cost.
For schema matching certainty, we choose entropy as our measure -- we are building our system on top of a basic schema-matching tool, which can estimate probabilities for schema matches it produces.  When the tool obtains a good match, it can associate a high probability.  When there is ambiguity or confusion, this translates into multiple lower probability matches, with associated uncertainty and hence higher entropy.

Our first algorithm, called Single CCQ (CCQ is short for Correspondence Correctness Question), determines the single most
valuable correspondence query to ask the crowd, given a set of
possible schema matchings and associated correspondences, all with
probabilities.

Intuitively, one may try a simple greedy approach, choosing the
query that reduces entropy the most.  However, there are three
issues to consider.  First, the correspondences are not all
independent, since they are related through candidate matchings.  So
it is not obvious that a greedy solution is optimal.  Second, even
finding the query that decreases entropy the most can be
computationally expensive.  Third, we cannot assume that every
person in the crowd answers every question correctly -- we have to
allow for wrong answers too.  We address all three challenges below.

Usually, we are willing to ask the crowd about more than one
correspondence, even if not all of them.   We could simply run
Single CCQ multiple times, each time greedily resolving uncertainty
in the most valuable correspondence.  However, we can do better. For
this purpose, we develop Multiple CCQ, an extension of Single CCQ,
that maintains $k$ most useful questions to ask the crowd, and
dynamically updates questions according to newly received answers.

In a previous conference paper \cite{DBLP:pvldb/ZhangCJC13}, we addressed this problem assuming crowds to be always correct. In this paper we consider more realistic situations: (1) Each CCQ has a probability to be answered correctly depending on the hardness of CCQ; (2) Each crowd worker has a probability to answer a CCQ correctly, which shows the trustworthiness of the worker. Therefore \cite{DBLP:pvldb/ZhangCJC13} can be viewed as a special case of our paper (probabilities all equal to 1). Combining above two situations together, we could compute the probabilities of CCQs to be answered correctly, and publish k CCQs chosen by our model to crowds with accuracy rates.

To summarize, we have made following contributions,

1. In Section~\ref{section:formalization} and Section~\ref{subsection:Multiple-CCQ-Form}, we propose an
entropy-based model to formulate the uncertainty reduction caused by
a single CCQ and multiple CCQs, respectively.

2. For the Single CCQ approach, we propose an explicit framework to choose a CCQ,
and derive an efficient algorithm in
Section~\ref{SCCQA}.  We introduce an index structure and pruning technique for efficiently
finding the Single CCQ.

3. In Section 3.3 and 4.3, we prove for both Single CCQ approach and Multiple CCQ approach that uncertainty reduction equals to entropy of answer minus entropy of crowds. In Section 3.3 we give the property of uncertainty reduction for Single CCQ approach. In Section 4.4 we obtain optimal upper and lower bounds for Multiple CCQ approach.

4. For the Multiple CCQ approach, we prove its NP-hardness in Section 4.5, and propose an
efficient $(1+\epsilon)$ approximation algorithm, with effective
pruning techniques in Section 4.6.

5. Section~\ref{Experiment} reports and discusses the
experimental study on both simulation and real implementation. We
review and compare our solutions with related work in
Section~\ref{related}. In Section~\ref{conclusion}, we conclude the paper and discuss future work.
\vspace{-3ex}

\section{Problem Statement}
\label{sec:problem}
\vspace{-1ex}
In this section, we give definitions related to the problem that
we are working on in this paper.\vspace{-1ex}
\begin{definition}[Correspondence]
Let \textit{S} and \textit{T} be two given schemata. A
correspondence $c$ is a pair $(A_s,A_t)$, where  $A_s$ and $A_t$ are
two subsets of attributes from \textit{S} and \textit{T} respectively.
\vspace{-1ex}
\end{definition}
\textbf{Remark:}
Here we consider correspondences between subsets of  \textit{S} and \textit{T}, which means $c$ could be not only 1:1 matching, but also n:m matching. For example in Table 1, $c_1$ is a 2:1 matching.
\vspace{-1ex}
\begin{definition}[Possible Matching]
Let \textit{S} and \textit{T} be two given schemata. Possible
matching $m_i=\{c_1,c_2,\dots,c_{|m_i|}\}$ is a set of correspondences between \textit{S}
and \textit{T} which satisfies that no attribute participate in more than one correspondence.
\vspace{-1ex}
\end{definition}
\textbf{Remark:}
For example, in Table~\ref{table:schemA}, $m_1$,$m_2$ and $m_3$ are
three possible matchings. Note that not every set of correspondences is a possible matching. In practice, possible matchings are generated by schema matching tools with probabilities to be correct.\vspace{-1ex}
\begin{definition}[Result Set]
\label{definition:RS}
For two given schemata \textit{S} and \textit{T}, let the result set
\textit{R} be the set of possible matchings generated by some
semi-automatic tool of schema matching, together with a probability
assignment function $\mathbb{P}:R\rightarrow [0,1]$. Each matching $ m_i
\in R$ has the probability $\mathbb{P}(m_i)$ to be correct, and we have $\sum_{m_i \in \textit{R}}{\mathbb{P}(m_i)}=1$
\vspace{-1ex}
\end{definition}
\textbf{Remark:}
In the example of Table 1, the set $\{m_1, m_2, m_3\}$ is result set. In practice, schema matching tools may use some threshold to eliminate possible matchings with very low probability, and return only a few higher probability candidates.  If such thresholding is performed, we ignore the low probability matchings that are already pruned, and set their probability to zero. We also mention that \cite{nguyen2014pay} discussed another way to establish the probability of each matching. \vspace{-1ex}
\begin{definition}[Correspondence Set]
\label{CS} Let \textit{R} be the result set for two given schemata
\textit{S} and \textit{T}, the correspondence set $C$ is the set of
all correspondences contained by possible matchings in $R$, i.e.
$C = \bigcup_{m_i\in R}m_i$
\vspace{-1ex}
\end{definition}

\textbf{Remark:}
Note that a correspondence can appear in more than one possible
matching, so for any correspondence $c\in C$, let $\mathbb{P}(c)$ be the
probability of $c$ being in the correct matching, then\vspace{-1ex}
\begin{sequation}
\mathbb{P}(c)=\sum_{\substack{m_{i}\in R\\c\in m_{i}}}\mathbb{P}(m_{i})\label{eq:CProb}
\vspace{-1.5ex}
\end{sequation}As a simple extension, for a set of correspondences $U \subseteq
C$, let $\mathbb{P}(U)$ be the probability that all correspondences of $U$
are in the correct matching, then\vspace{-1ex}
\begin{sequation}
\mathbb{P}(U)=\sum_{\substack{m_{i}\in R\\U\subseteq m_{i}}}\mathbb{P}(m_{i})\label{eq:SetProb}
\vspace{-2ex}
\end{sequation}
For example in Table 1, $\{c_1, c_2, c_3, c_4, c_5\}$ is correspondence set. Since $c_1$ is in $m_1$ and $m_2$, $\mathbb{P}(c_1)=0.75$.
\vspace{-1ex}
\begin{definition}[Uncertainty of Schema Matching]
For two given schemata \textit{S} and \textit{T}, given result set \textit{R} and probability assignment function $\mathbb{P}()$, we measure the uncertainty of \textit{R} with Shannon entropy
$H(R) = -\sum_{m_i \in \textit{R}}{\mathbb{P}(m_i)}\log{\mathbb{P}(m_i)}$
\vspace{-1ex}
\end{definition}

\textbf{Remark:}
Shannon entropy has been widely used in information theory and many other fields since 1948 \cite{Shannon} to measure the uncertainty, disorder or unpredictability of a system. Another way to measure uncertainty is to use variance or covariance matrix. For some special cases like Bernoulli distribution, Shannon entropy has maximal value or minimal value when its variance is maximal or minimal.  A major reason for utilizing Shannon entropy as is its non-parametric nature. The probability distribution of possible matchings is very dynamic,
depending not only on the given schemata, but also on the schema
matching tools. Entropy does not require any assumptions about the
distribution of variables. Besides, entropy permits non-linear
models, which is important for categorical variables
\cite{lemay1999statistical}, such as possible matchings. \vspace{-1ex}
\nop{In
addition, considering schema matching as an uncertain query, it can
be argued  \cite{DBLP:journals/pvldb/ChengCX08} that entropy is an
accurate measurement for the quality of uncertain queries.}

\begin{definition}[Crowd's Accuracy]
Given a crowd worker W, the
crowd's accuracy (or accuracy for short), denoted by $P_{W}\in[0.5,1]$,  is the probability that W correctly answers each HIT.
\vspace{-1ex}
\end{definition}
\textbf{Remark:}
While some
papers assume that crowdsourced answers are $100\%$ accurate, we adopt a more general error
model, which requires only that the answer returned by each crowd worker is
always correct with a probability no lower than $1/2$. This is a classical crowdsourcing model
widely used by a stream of works \cite{DBLP:conf/icdt/DavidsonKMR13},  \cite{DBLP:conf/sigmod/GuoPG12},  \cite{DBLP:pvldb/LiuLOSWZ12}, \cite{DBLP:conf/sigmod/ParameswaranGPPRW12}.
Crowd workers may have different accuracies for different domains.
The accuracy for a domain can be easily estimated with a set
of sample HITs in which ground truth is known. Before we ask a CCQ, we could assume that this CCQ will be answered by a crowd with accuracy $P_W$. Since we do not know who will answer this CCQ, $P_W$ is likely to represent the hardness of CCQ as an attribute of the correspondence.\vspace{-1ex}

\begin{definition}[Entropy of Crowd]
Given a crowd worker W and
its accuracy $P_{W}$ , the entropy of W is defined by\vspace{-1ex}
\begin{sequation}
H(W)=-P_{W}\log P_{W}-\left(1-P_{W}\right)\log\left(1-P_{W}\right)\label{eq:CrowdEntropy}
\vspace{-1.5ex}
\end{sequation}Given the crowd's accuracy, $H(W)$ is a positive constant
measuring the randomness of the crowd's behaviour.
\vspace{-1ex}
\end{definition}

\begin{definition}[\small{Correspondence Correctness Question}]
A Correspondence Correctness Question (CCQ) asks whether a
correspondence is correct. The CCQ w.r.t a correspondence $c$ is
denoted as $Q_{c}$, where $c\in C$.
\vspace{-1.5ex}
\end{definition}

\textbf{Remark:}
A running example where we calculate entropy and conditional probability after we have an answer of CCQ is given in Section 3.3.1.
\vspace{-1ex}
\begin{definition}[Entropy of Answer]
\label{definition: EntropyA}
Given result set \textit{R}, probability assignment function $\mathbb{P}()$, and crowd's accuracy $P_{W}\in[0.5,1]$, the entropy of answer A corresponding to question $Q_c$ is defined by\vspace{-1ex}
\begin{sequation}\label{Eq:EntropyOfA}
H(A)=-\mathbb{P}(A=Y)\log\mathbb{P}(A=Y)-\mathbb{P}(A=N)\log\mathbb{P}(A=N)
\vspace{-1.5ex}
\end{sequation}where
\begin{sequation}\label{eq:AYN}
\begin{aligned}
\mathbb{P}(A=Y) & =\mathbb{P}(c)P_{W}+\left(1-\mathbb{P}(c)\right)\left(1-P_{W}\right)\\
\mathbb{P}(A=N) & =\left(1-\mathbb{P}(c)\right)P_{W}+\mathbb{P}(c)\left(1-P_{W}\right)
\end{aligned}
\end{sequation}
\vspace{-4ex}
\end{definition}

\begin{definition}[Problem Statement]
On two given schemata \textit{S} and \textit{T}, let the result set
\textit{R} and probability assignment function $\mathbb{P}$ be generated by
some schema matching tools. Each CCQ is assumed to be answered independently. Let \textit{B} be the budget of the
number of CCQs to be asked to the crowd. Our goal is to maximize the
reduction of $H(R)$ without exceeding the budget.
\end{definition}
\vspace{-4.2ex}

\section{Single CCQ Approach}
\label{SCCQA}
\vspace{-1ex}
In this section, we study how to choose a single CCQ well.
To be able to do this, we first address the formalization of uncertainty reduction.
Then we develop the {\em Single CCQ Approach}, a framework to address the uncertainty reduction problem using a sequence of Single CCQ. Compared with \cite{DBLP:pvldb/ZhangCJC13}, we give a new proof for uncertainty reduction under the condition of accuracy probability $P_W \in [0.5,1]$. We prove the equivalent form of uncertainty reduction and its property. Finally, we propose efficient
algorithms to implement the computations in this approach.\vspace{-3.1ex}
\begin{table}[!t]\caption{\footnotesize{Meanings of Symbols Used}} 
\label{symbols} \footnotesize
\centering 
\vspace{-3ex}
\begin{tabular}{c c} 
\hline 
Notation & Description \\ 
$c_i$, $C$, $Q_{c_i}$  & correspondence, correspondence set, CCQ w.r.t $c_i$  \\
$m_i$, $\left|m_{i}\right|$  & a possible matching, number of elements in $m_i$ \\
$\mathbb{P}(c_i)$ & probability of $c_i$ being in the correct matching  \\
$\mathbb{P}(m_i)$ & probability that $m_i$ is the correct matching  \\
$A$  or $A_{c_i}$ & the answer or the answer for correspondence $c_i$ \\
$R$, $W$, $P_W$  & result set, a crowd worker, crowd's accuracy  \\
$H(R)$, $H(W)$, $H(A)$& entropy of result set, crowd, answer \\
$\Delta H_{Q_c}$ & uncertainty reduction by publishing $Q_{c}$ \\
$D_A$ & the domain of answers of k CCQs  \\
$H(D_A)$ & joint entropy of k answers\\
$D_U$ & the domain of k correspondences\\
\vspace{-8ex}

\end{tabular} 
\end{table}

\subsection{Formulation of Uncertainty Reduction}
\vspace{-1ex}
\label{section:formalization} In order to design an effective
strategy for manipulating CCQs, it is essential to define a
measurement to estimate the importance of CCQs before they are
answered. Since the final objective is to reduce uncertainty, we use
 uncertainty reduction caused by individual CCQs
as the measurement. In the following, we provide the formulation of
the  uncertainty reduction in the context of the Single CCQ
Approach.

\nop{At this moment, we first focus on the case when only one CCQ is
to be selected. The results in this subsection are illustrative when
considering extensions to handle the case of multiple CCQs, as we
shall show in Section~\ref{subsection:Multiple-CCQ-Form}.}

Let $Q_c$ be a CCQ w.r.t an arbitrary correspondence $c$. We assume
crowdsourcing workers provide answers independently with accuracy $P_W \in [0.5,1]$. Since
$Q_c$ is a Yes/No question, we consider the answer $A$ as a random variable
following a Bernoulli distribution. Firstly, we have $D_{A}=\left\{ Y,N\right\}$ and $\mathcal{P}=\left\{\mathbb{P}\left(A=Y\right),\mathbb{P}\left(A=N\right)\right\}$, where $\mathbb{P}(A=Y)$ and $\mathbb{P}(A=N)$ can be computed by Eq~\ref{eq:AYN}. We write $\mathbb{P}(Y)$ and $\mathbb{P}(N)$ for short.
For two discrete random variables $X$ and $Y$ with p.m.f. function $p$, the conditional entropy is defined by\vspace{-1ex}
\begin{footnotesize}
\begin{align*}
H(Y|X) & =\sum_{x\in\mathcal{X}}p_{X}(x)H(Y|X=x)\\
 & =-\sum_{x\in\mathcal{X}}p_{X}(x)\sum_{y\in\mathcal{Y}}p_{Y}(y|X=x)\log p_{Y}(y|X=x)
\vspace{-1ex}
\end{align*}
\end{footnotesize}
Let $\Delta H_{Q_c}$ be the uncertainty reduction caused by $Q_c$, we have
\vspace{-1ex}
\begin{sequation}
\bigtriangleup H_{Q_{c}} =H(R)-H\left(R|A\right)
\label{eq:Ed}
\vspace{-1.2ex}
\end{sequation}
where\vspace{-1ex}
\begin{footnotesize}
\begin{align*}
-H\left(R|A\right) & =\mathbb{P}(Y)\sum_{m_{i}\in R}\left[\mathbb{P}\left(m_{i}|Y\right)\log\mathbb{P}\left(m_{i}|Y\right)\right]\\
 & \quad+\mathbb{P}(N)\sum_{m_{i}\in R}\left[\mathbb{P}\left(m_{i}|N\right)\log\mathbb{P}\left(m_{i}|N\right)\right]
\end{align*}
\end{footnotesize}
\vspace{-1ex}
\begin{sequation}
\begin{aligned}\mathbb{P}(m_{i}|Y) & =\frac{\mathbb{P}(m_{i})\mathbb{P}(Y|m_{i})}{P_{W}\mathbb{P}(c)+\left(1-P_{W}\right)(1-\mathbb{P}(c))}\\
\mathbb{P}(m_{i}|N) & =\frac{\mathbb{P}(m_{i})\mathbb{P}(N|m_{i})}{P_{W}(1-\mathbb{P}(c))+\left(1-P_{W}\right)\mathbb{P}(c)}
\end{aligned}
\label{eq:adjust}
\vspace{-1ex}
\end{sequation}The uncertainty reduction w.r.t a given $Q_c$ can be computed by Eq~\ref{eq:Ed}
provided that we know the values for parameters:
$\mathbb{P}(c)$, $\mathbb{P}(m_i)$, $\mathbb{P}(Y|m_i)$ and $\mathbb{P}(N|m_i)$. $\mathbb{P}(c)$ can be computed by Eq~\ref{eq:CProb}.
$\mathbb{P}(Y|m_i)$ and $\mathbb{P}(N|m_i)$ depend on
if $c$ is a correspondence included in $m_i$.\vspace{-1ex}
\begin{sequation}
\begin{aligned}
\label{accuracyConditionedonPR}
\mathbb{P}(Y|m_i)=\begin{cases}P_W&c\in m_i\\ 1-P_W&c\notin
m_i\end{cases};
\\
\mathbb{P}(N|m_i)=\begin{cases}1-P_W&c\in m_i\\ P_W&c\notin
m_i\end{cases}
\end{aligned}
\vspace{-2ex}
\end{sequation}

\textbf{Remark: The harmlessness of random answer} If a worker W randomly answers a CCQ $Q_c$, i.e. $P_W=0.5$ and $\mathbb{P}(Y)=0.5$, it does not affect the uncertainty of schema matching. In other word, by Eq~\ref{eq:adjust}, we have $\mathbb{P}(m_i|Y)=\mathbb{P}(m_i)$.

Eq~\ref{eq:adjust} is applied recursively as multiple answers are received, to take all of them into account.  If multiple answers all agree, each iteration will make the truth of $c$ more certain, whereas disagreeing answers will pull the probability closer to the middle.  In other words, disagreements between workers are gracefully handled.
It is easy to perform the algebraic manipulations to show
that, for any two answers $A_1$ and $A_2$, we have $\mathbb{P}(m_i|A_1,A_2)=\mathbb{P}(m_i|A_2,A_1)$.
This equation indicates that the result of adjustment is
\textbf{independent} of the sequence of the answers.  In other words, when we have a deterministic set of questions (CCQs), it does not matter in what sequence the answers are used for adjustment. In contrast, what matters is to determine the set of CCQs to be asked, which is the core challenge addressed in this paper.\vspace{-3ex}

\subsection{Framework of Single CCQ}\vspace{-1ex}
Having developed a technique to find the best Single CCQ, we can place this at the heart of an approach to solve the schema matching problem, as shown in Framework~\ref{SingleCCQ}.
The idea is to greedily select the single CCQ in each iteration that will result in the greatest reduction of uncertainty.
We publish this CCQ; when it is answered and returned with accuracy rate
(line 4), we adjust $\mathbb{P}(m_i)$ by $\mathbb{P}(m_i|A)$ (line 5), and then generate a new CCQ (line 7\&8).

\begin{salgorithm}
\floatname{algorithm}{\small{Framework}} \caption{\small{Single CCQ}} \label{SingleCCQ}
\begin{algorithmic}[1]
 \State $CONS\gets 1$ // consumption of the budget
 \State Find and publish $Q_{c_i}$ that maximize $\Delta H_{Q_c} $ //
   \While {there exists a CCQ in the crowd, we constantly monitor the CCQ}
\For{answer $A_i$ of $Q_{c_i}$, accuracy rate $P_{W_i}$}
   \State $\forall m_i\in R$ Adjust the $\mathbb{P}(m_i)$ by $\mathbb{P}(m_i|A)$  //
 \If {$CONS<B$}
 \State Finding $Q_{c_j}$ maximizing $\Delta H_{Q_c}$ //
\State publish $Q_{c_j}$, \State $CONS = CONS+1$ \Else  \State
terminate (no more budget)
 \EndIf
\EndFor
\EndWhile\vspace{-1ex}
\end{algorithmic}
\end{salgorithm}

In the framework of Single CCQ, one can see that an important task is to
find the CCQ with the highest uncertainty reduction
as soon as the probability distribution of $R$ is adjusted (line
7). We can formally pose this as a query as follows, and focus on efficiently processing such a query in the rest of this section.\vspace{-1ex}
\begin{definition}[Single CCQ Selection (SCCQS)]
Given result set \textit{R}, probability assignment function $\mathbb{P}()$, crowd's accuracy $P_W \in[0.5,1]$,
the Single CCQ Selection Query retrieves a CCQ maximizing the
 uncertainty reduction  $ \Delta H_{Q_c}$ in Eq~\ref{eq:Ed}.
\end{definition}
\vspace{-4ex}
\subsection{Query Processing of SCCQS}
\vspace{-1ex}
Based on the formulation in Section~\ref{section:formalization}, we
are able to compute the uncertainty reduction of each CCQ.
So a naive approach of selection is to traverse all the CCQs. Such
traversal results in an algorithm with time complexity $\mathcal{O}(|R|^2|C|)$, i.e.
the square of the number of possible matchings multiplied by the
number of correspondences.  This can be a very large number for complex schema.

In this subsection, we first provide a lossless simplification, by
proving the uncertainty reduction is mathematically
equivalent to the entropy of the answer of a CCQ minus the entropy of the crowd. Then, in order to
further improve the efficiency, we propose an index structure based
on binary coding, together with a pruning technique.\vspace{-2ex}
\subsubsection{Simplification of Single CCQ Selection}\vspace{-1ex}
\label{proofSCCQ} When we need to determine a strategy of selecting CCQs, a
very intuitive idea is to prioritize the ones that we are more
uncertain. In case of Single CCQ, this idea suggests that we select the
CCQ with probability closest to $0.5$. This idea is trivially correct when all the correspondences are independent. However, with the model of possible matchings, there are correlations among the correspondences. Then, a non-trivial question is: should we still pick the CCQ with probability closest to $0.5$ with the presence of correlation?

Interestingly, we discover
that the answer is positive. By
Theorem~\ref{Simp-Single}, we prove that the  uncertainty reduction
$\Delta H_{Q_c}$ of a correspondence $c$ is equivalent to the entropy
of the answer $H(A)$ minus the entropy of the crowd $H(W)$. In other words,for a fixed $P_W$, $\Delta H_{Q_c}$ is only determined by $\mathbb{P}(c)$. As a result, searching for the CCQ that maximize $\Delta
H_{Q_c}$ has the complexity decreased to $\mathcal{O}(|R||C|)$, by
computing $\mathbb{P}(c)$ for each $c\in C$. In addition,
Theorem~\ref{Simp-Single-c} states that we only need to find the
correspondence that has probability closest to $0.5$, based on the
fact that $\Delta H_{Q_c}$ is a symmetric function of $\mathbb{P}(c)$, with
symmetry axis $\mathbb{P}(c)=0.5$ and achieves maximum when $\mathbb{P}(c)=0.5$.\vspace{-1ex}

\begin{theorem} \label{Simp-Single}
 Given result set \textit{R}, probability assignment function $\mathbb{P}()$, crowd's accuracy $P_W \in[0.5,1]$, for correspondence $c\in C$, we have\vspace{-2ex}
 \begin{footnotesize}
 \begin{align*}
 \Delta H_{Q_c} = H(A)-H(W)
 \vspace{-2ex}
\end{align*}
\end{footnotesize}where we recall $\Delta H_{Q_c}$ in Eq~\ref{eq:Ed} and  $H(A)$, $H(W)$ are defined in Eq~\ref{Eq:EntropyOfA}, Eq~\ref{eq:CrowdEntropy}.
\end{theorem}
\vspace{-2ex}
\begin{proof}
Using Eq~\ref{eq:adjust} into Eq~\ref{eq:Ed}, we have
\vspace{-1ex}
\begin{footnotesize}
\begin{align*}
\allowdisplaybreaks
 & \bigtriangleup H_{Q_{c}}\\
 & =H(R)+\sum_{m_{i}\in R}\left[\mathbb{P}(m_{i})\mathbb{P}(Y|m_{i})\log\mathbb{P}(m_{i})\right]\\
 & \quad+\sum_{m_{i}\in R}\left[\mathbb{P}(m_{i})\mathbb{P}(Y|m_{i})\log\mathbb{P}\left(Y|m_{i}\right)-\mathbb{P}(m_{i})\mathbb{P}(Y|m_{i})\log\mathbb{P}(Y)\right]\\
 & \quad+\sum_{m_{i}\in R}\left[\mathbb{P}(m_{i})\mathbb{P}(N|m_{i})\log\mathbb{P}(m_{i})+\mathbb{P}(m_{i})\mathbb{P}(N|m_{i})\log\mathbb{P}\left(N|m_{i}\right)\right]\\
 & \quad-\sum_{m_{i}\in R}\left[\mathbb{P}(m_{i})\mathbb{P}(N|m_{i})\log\mathbb{P}(N)\right]\\
 & \coloneqq H(R)+J_{1}+J_{2}+J_{3}+J_{4}+J_{5}+J_{6}
\end{align*}\vspace{-1.5ex}
\end{footnotesize} By Eq~\ref{accuracyConditionedonPR}, we obtain that
\begin{footnotesize}
\[
J_{1}+J_{4}=\sum_{m_{i}\in R}\mathbb{P}(m_{i})\log\mathbb{P}(m_{i})=-H(R)
\]
\end{footnotesize}
and
\vspace{-1.5ex}
\begin{footnotesize}
\begin{align}
 & J_{2}+J_{5}\nonumber \\
 & =\sum_{\substack{m_{i}\in R\\
c\in m_{i}
}
}\mathbb{P}(m_{i})P_{W}\log P_{W}+\sum_{\substack{m_{i}\in R\\
c\notin m_{i}
}
}\mathbb{P}(m_{i})\left(1-P_{W}\right)\log\left(1-P_{W}\right)\nonumber \\
 & \quad+\sum_{\substack{m_{i}\in R\\
c\notin m_{i}
}
}\mathbb{P}(m_{i})P_{W}\log P_{W}+\sum_{\substack{m_{i}\in R\\
c\in m_{i}
}
}\mathbb{P}(m_{i})\left(1-P_{W}\right)\log\left(1-P_{W}\right)\nonumber \\
 & =-H(W)\vspace{-4.5ex}
 \label{eq:j2j5}
\end{align}
\end{footnotesize}Recall Eq~\ref{eq:CProb}. It follows that\vspace{-1ex}
\begin{footnotesize}
\begin{align*} & J_{3}+J_{6}\\
 & =-\sum_{\substack{m_{i}\in R\\
c\in m_{i}
}
}\mathbb{P}(m_{i})P_{W}\log\mathbb{P}(Y)-\sum_{\substack{m_{i}\in R\\
c\notin m_{i}
}
}\mathbb{P}(m_{i})\left(1-P_{W}\right)\log\mathbb{P}(Y)\\
 & \quad-\sum_{\substack{m_{i}\in R\\
c\notin m_{i}
}
}\mathbb{P}(m_{i})P_{W}\log\mathbb{P}(N)-\sum_{\substack{m_{i}\in R\\
c\in m_{i}
}
}\mathbb{P}(m_{i})\left(1-P_{W}\right)\log\mathbb{P}(N)\\
 & =-\left[\mathbb{P}(c)P_{W}+\left(1-\mathbb{P}(c)\right)\left(1-P_{W}\right)\right]\log\mathbb{P}(Y)\\
 & \quad-\left[\left(1-\mathbb{P}(c)\right)P_{W}+\mathbb{P}(c)\left(1-P_{W}\right)\right]\log\mathbb{P}(N)\vspace{-1ex}\\
 & =H(A)\vspace{-1ex}
\end{align*}\vspace{-1.5ex}
\end{footnotesize}This completes the proof.
\end{proof}

\begin{theorem}
\label{Simp-Single-c} The uncertainty reduction of single CQQ is always non-negative for $P_W \in [0.5,1]$. For any two correspondence $c,c' \in C$,  if
$|0.5-\mathbb{P}(c)|\leq|0.5-\mathbb{P}(c')|$ then
 $\Delta H_{Q_c} \geq \Delta H_{Q_c'} \geq 0$. In addition, $\Delta H_{Q_c}=1-P_W$ if $\mathbb{P}(c)=0.5$
\end{theorem}
\vspace{-2.5ex}
\begin{proof}
By Theorem~\ref{Simp-Single}, $\Delta H_{Q_c}$ is a function of $\mathbb{P}(Y)$ and $P_W$. $\mathbb{P}(Y)$ is a function of $\mathbb{P}(c)$ and $P_W$. We first consider\vspace{-1ex}
\begin{footnotesize}
\[
\frac{\partial\bigtriangleup H_{Q_c}}{\partial\mathbb{P}(Y)}=-\log\frac{\mathbb{P}(Y)}{1-\mathbb{P}(Y)}\vspace{-1ex}
\]
\end{footnotesize}We could obtain that\vspace{-1ex}
\begin{footnotesize}
\[
-\log\frac{\mathbb{P}(Y)}{1-\mathbb{P}(Y)}=0\Leftrightarrow\mathbb{P}(Y)=0.5\vspace{-1ex}
\]
\end{footnotesize}It is easy to check that $\Delta H_{Q_c}$ is a symmetric function of $\mathbb{P}(Y)$, with symmetry axis
$\mathbb{P}(Y)=0.5$. Besides, the function achieves maximum $\Delta H_{Q_c}=1-H(W)$ when $\mathbb{P}(A=Y)=0.5$, and is monotonic on $[0,0.5]$
(increasing) and $[0.5,1]$ (decreasing). We also know that $\mathbb{P}(Y)=\mathbb{P}(c)P_{W}+\left(1-\mathbb{P}(c)\right)\left(1-P_{W}\right)$. So $\mathbb{P}(Y)$ is increasing w.r.t. $\mathbb{P}(c)$ and achieves the value $0.5$ when $\mathbb{P}(c)=0.5$ or $P_{W}=0.5$. Thus $\Delta H_{Q_c}$ achieves maximum $1-H(W)$ when $\mathbb{P}(c)=0.5$. Secondly we consider\vspace{-1ex}
\begin{footnotesize}
\[
\frac{\partial\bigtriangleup H_{Q_{c}}}{\partial P_{W}} =\left(2\mathbb{P}(c)-1\right)\log\frac{\mathbb{P}(c)-\left(2\mathbb{P}(c)-1\right)P_{W}}{\left(2\mathbb{P}(c)-1\right)P_{W}-\mathbb{P}(c)+1}+\log\frac{P_{W}}{1-P_{W}}
\vspace{-1ex}
\]
\end{footnotesize}Since $\Delta H_{Q_c}$ is a symmetric function of $\mathbb{P}(A=Y)$, with symmetry axis
$\mathbb{P}(A=Y)=0.5$ and $\mathbb{P}(A=Y)$ is increasing w.r.t. $\mathbb{P}(c)$, we choose  $\mathbb{P}(c)=0$  and $\mathbb{P}(c)=1$ in order to obtain minimum of $\Delta H_{Q_c}$. When $\mathbb{P}(c)=0$  or $\mathbb{P}(c)=1$, we have
$
\frac{\partial\bigtriangleup H_{Q_c}}{\partial P_{W}}=0
$
and $\Delta H_{Q_c}=0$. Thus we prove that $\Delta H_{Q_c}$ is non-negative.\vspace{-1.5ex}
\end{proof}

\textbf{Running Example (Selecting First Two CCQs):} Now we illustrate the process of selecting the first two CCQs in Framework~\ref{SingleCCQ} with the example of Table~\ref{table:schemB}. In line 2, the first correspondence to be asked is $c_2$, since its probability is closest to 0.5 among $\{c_1,c_2,c_3,c_4,c_5\}$. Explicitly, $\Delta H_{Q_{c_2}}= -0.7*log(0.7)-0.3*log(0.3)=0.88$. Suppose an answer ``$a=yes$'' is received from a crowd worker, whose personal error rate is $1-P_W=0.2$ (line 4). Then we conduct the adjustment according to Eq~\ref{eq:adjust}, and have $\mathbb{P}(m_1|a)=0.58$, $\mathbb{P}(m_2|a)=0.10$ and $\mathbb{P}(m_3|a)=0.32$. This adjustment is referring to the first-time execution of line 5. Then, in line 7, the next CCQ is to be selected. Note that, since the probabilities of possible matchings are adjusted, probabilities of correspondences should be recomputed by Eq~\ref{eq:CProb}: $\mathbb{P}(c_1)=0.68$, $\mathbb{P}(c_2)=0.9$, $\mathbb{P}(c_3)=1$, $\mathbb{P}(c_4)=0.68$, $\mathbb{P}(c_5)=0.32$. Therefore, in line 7, we select the CCQ based on the updated probabilities of correspondences, i.e. $c_1$ would be selected. (There is a tie among $c_1$,$c_4$ and $c_5$, and we break the tie sequentially.)\vspace{-1ex}

\subsubsection{Binary Coding and Pruning Techniques}
\vspace{-1ex}
One can see that a basic computation of our algorithm is to check
whether a given correspondence $c$ is in a given possible matching
$m_i$. Since the correspondences included in each possible matching do not change
with the value of overall uncertainty, we propose to index $R$
with a binary matrix $M_{R}$, where element $e_{ij}=1(0)$
representing $c_j\in m_i (c_j\notin m_i)$. Equipped with this index, we apply a pruning technique derived from
Theorem~\ref{Simp-Single-c}.

Now we illustrate the procedure of generating the correspondence with probability
closest to 0.5. For each $c_j$, we traverse $m_i$ and accumulate
$\mathbb{P}(m_i)$ if $e_{ij}=1$. Let $c_{best\_so\_far}$ be the best
correspondence so far, with probability $\mathbb{P}(c_{best\_so\_far})$.
Then, let  $c_j$ be the current correspondence, and $P_{acc}$ be its
accumulated probability after reading some $\mathbb{P}(m_i)$, then $c_j$ can
be safely pruned if we have $ P_{acc} - 0.5 \geq
|\mathbb{P}(c_{best\_so\_far})-0.5|.$ \vspace{-2.6ex}

\section{Multiple CCQ Approach}
\label{MCCQA}
\vspace{-1ex}
A drawback of single CCQ is that only one correspondence is resolved at a time.
Each resolution, even if quick, requires human time scales, and comes with some overhead to publish the corresponding HIT and tear it down.  Gaining confidence in a single schema matching may require addressing many CCQs.  The time required to do this in sequence may be prohibitive.

An alternative we consider in this section is to issue multiple ($k$) CCQs simultaneously.
Different workers can then pick up these tasks and solve them in parallel, cutting down wall-clock time.  However, we pay for this by having some questions answered that are not at the top of the list -- we are issuing $k$ good questions rather than only the very best one.

Note that there are three possible states for a published CCQ: (1)
\textbf{waiting} - no one has accepted the question yet; (2)
\textbf{accepted} - someone in the crowd has accepted the question
and is working on it; (3) \textbf{answered} - the answer of the CCQ
is available. What's more important, one can withdraw published CCQs
that are still at state waiting  (e.g. forceExpireHIT in Mechanical
Turk APIs)\cite{DBLP:conf/sigmod/FranklinKKRX11}. In other words,
publishing a CCQ does not necessarily consume the budget. It is
possible that a CCQ is published, and then withdrawn before anyone
in the crowd answers it. In such case, the budget is not consumed.
Because of the dependence between correspondences, we can withdraw or
replace some of the published CCQs that are at ``waiting'' state.
Equipped with this power, we propose the Multiple CCQ approach to
dynamically keep $k$ best CCQs published at all times.

In the rest of this section, we provide the formulation and
framework of Multiple CCQs, by extending our results of Single CCQ.
Compared with our conference paper \cite{DBLP:pvldb/ZhangCJC13}, we give new proofs for uncertainty reduction under more general condition that crowd workers have accuracy probabilities $P_{W_i} \in [0.5,1]$. These probabilities can show hardness of CQQs or how professional workers are. Accuracy probabilities are assumed before we ask CCQs and are returned with answers after we publish CCQs. They can be totally different for different correspondences and different crowd workers.  We prove \textbf{the uncertainty reduction equals to joint entropy of answers minus sum of entropies of crowds}. We also show upper and lower bounds for uncertainty reduction. Results in \cite{DBLP:pvldb/ZhangCJC13} can be viewed as a special case when workers are always correct. Finally, we prove the NP-hardness of the
multiple CCQs selection problem, and propose an efficient
approximation algorithm with bounded error.\vspace{-3.2ex}

\subsection{Formulating Uncertainty Reduction of Multiple CCQ Approach}
\label{subsection:Multiple-CCQ-Form}

For a set of CCQs of size k - $S_{Q} = \{Q_{c_1},Q_{c_2},...,Q_{c_k}\}$, $A_{c_{1}}$, $A_{c_{2}}$, ..., $A_{c_{k}}$ denote answers of k CCQs given by k workers $W_1$, $W_2$, ..., $W_k$ with accuracy $P_{W_1}$, $P_{W_2}$, ..., $P_{W_k}$. We want to derive the uncertainty reduction caused by the aggregation of the answers of these k CCQs. Let $D_A$ and $\mathcal{P}_A$ be the domain and probability distribution of answers respectively.  Each element of $D_A$ is a possible set of answers for k CCQs (a sequence of Y and N) with a corresponding probability in $\mathcal{P}_A$. Then first we have\vspace{-1ex}
\begin{footnotesize}
\begin{align*}
D_{A} & =\left\{ a_{i}\left|a_{i}=\left\{ A_{c_{1}}^{(i)},A_{c_{2}}^{(i)},...,A_{c_{k}}^{(i)}\right\} ,\;A_{c_{j}}^{(i)}=Y\;\mathrm{or}\;N\right.\right\} \\
\mathcal{P}_A & =\left\{\mathbb{P}\left(a_{1}\right),\mathbb{P}\left(a_{2}\right),...,\mathbb{P}\left(a_{2^{k}}\right)\right\}
\vspace{-1.5ex}
\end{align*}
\end{footnotesize}where $\left|D_{A}\right|=2^{k}$. As we know, each correspondence $c_t$ has a probability to show its ground truth, i.e. with $\mathbb{P}(c_t)$ to be true before crowds answer CCQs. We view $U=\{c_{t};t=1,...,k\}$ as a set of random variables which follow Bernoulli distribution and take value True/False. Note that they are not independent and their joint p.m.f. can be calculated by Eq~\ref{eq:SetProb}.   Let $D_U$ and $\mathcal{P}_U$ be the domain and probability distribution of $\{c_{t};t=1,...,k\}$ respectively.  Each element of $D_U$ is a sequence of T and F with a corresponding probability in $\mathcal{P}_U$. Thus we have\vspace{-1ex}
\begin{sequation}
\begin{aligned}D_{U} & =\left\{ u_{i}\left|u_{i}=\left\{ c_{1}^{(i)},c_{2}^{(i)},...,c_{k}^{(i)}\right\} ,\;c_{t}^{(i)}=T\; \mathrm{or}\; F\right.\right\} \\
\mathcal{P}_U & =\left(\mathbb{P}\left(u_{1}\right),\mathbb{P}\left(u_{2}\right),...,\mathbb{P}\left(u_{2^{k}}\right)\right)
\vspace{-1ex}
\end{aligned}
\vspace{-1ex}
\label{eq:DUandP}
\end{sequation}
\vspace{-1ex}
where $\left|D_{U}\right|=2^{k}$. By Eq~\ref{eq:SetProb}, we have
\begin{footnotesize}
\begin{equation}
\mathbb{P}\left(u_{i}\right)=\sum_{\substack{m_{j}\in R\\
t=1,...,k\\
\forall c_{t}^{(i)}=T,\;c_{t}\in m_{j}\\
\forall c_{t}^{(i)}=F,\;c_{t}\notin m_{j}
}
}\mathbb{P}(m_{j})\label{eq:Pui}
\vspace{-1.5ex}
\end{equation}
\end{footnotesize}
\textbf{Remark: Complexity} We remark that in computation of all $\mathbb{P}(u_{i})$, $i=1,...,2^k$, each $\mathbb{P}(m_{j})$, $j=1,...,|R|$ will be used once and only once. Therefore, the number of elements with positive probability in $D_U$ is less than or equal to $|R|$ and time complexity of computing all $\mathbb{P}(u_{i})$ is bounded by $\mathcal{O}(k|R|)$.

Similar to Eq~\ref{eq:Ed}, we are able to to compute the  uncertainty reduction caused by the $S_{Q}$, denoted by $\Delta H_{S_{Q}}$. We have\vspace{-1ex}
\begin{footnotesize}
\begin{align}
\Delta H_{S_{Q}} & =H(R)-H\left(R\left|A_{c_{1}},...,A_{c_{k}}\right)\right)\label{eq:multUD_exp}\\
 & =H(R)+\sum_{a_{i}\in D_{A}}\mathbb{P}(a_{i})\sum_{m_{j}\in R}\left[\mathbb{P}(m_{j}|a_{i})\log\mathbb{P}(m_{j}|a_{i})\right]\nonumber \\
 & =H(R)+\sum_{\substack{m_{j}\in R\\
a_{i}\in D_{A}
}
}\mathbb{P}(m_{j})\mathbb{P}(a_{i}|m_{j})\log\frac{\mathbb{P}(m_{j})\mathbb{P}(a_{i}|m_{j})}{\mathbb{P}(a_{i})}\nonumber
\end{align}
\vspace{-1.5ex}
\end{footnotesize}

\textbf{Computation of $\mathbb{P}(a_i)$}: For one CQQ $c_i$, $A_{c_i}$ is the answer given by a worker with accuracy $P_{W_i}$. When $A_{c_i}$ is Yes, $c_i$ may be True and worker is correct, or $c_i$ is False and worker is incorrect. It is easy to see that\vspace{-1ex}
\begin{footnotesize}
\begin{align*}
\mathbb{P}(A_{c_i}=Y) & =\mathbb{P}(c_i)P_{W_i}+\left(1-\mathbb{P}(c_i)\right)\left(1-P_{W_i}\right)\\
\mathbb{P}(A_{c_i}=N) & =\left(1-\mathbb{P}(c_i)\right)P_{W_i}+\mathbb{P}(c_i)\left(1-P_{W_i}\right)\vspace{-2ex}
\end{align*}
\end{footnotesize}For $k$ CQQs, the answers in $a_i$ are denoted by $A_{c_{1}}^{(i)}$, $A_{c_{2}}^{(i)}$, ..., $A_{c_{k}}^{(i)}$. Similarly, $\mathbb{P}(a_i)$ can be computed by Eq~\ref{eq:Pui}.\vspace{-1.5ex}
\begin{sequation}
\mathbb{P}(a_{i})=\sum_{j=1}^{2^{k}}\mathbb{P}(u_{j})q_{ij}
\label{eq:2p1}
\vspace{-1.5ex}
\end{sequation}where\vspace{-1ex}
\begin{footnotesize}
\[
q_{ij}=\prod_{\substack{t=1,...,k\\
c_{t}^{(j)}=T\\
A_{c_{t}}^{(i)}=Y
}
}P_{W_{t}}\prod_{\substack{t=1,...,k\\
c_{t}^{(j)}=T\\
A_{c_{t}}^{(i)}=N
}
}\left(1-P_{W_{t}}\right)\prod_{\substack{t=1,...,k\\
c_{t}^{(j)}=F\\
A_{c_{t}}^{(i)}=Y
}
}\left(1-P_{W_{t}}\right)\prod_{\substack{t=1,...,k\\
c_{t}^{(j)}=F\\
A_{c_{t}}^{(i)}=N
}
}P_{W_{t}}
\]
\end{footnotesize}\textbf{Computation of $\mathbb{P}(a_i|m_j)$}:  Similar to Single CCQ, Eq~\ref{accuracyConditionedonPR}, $\mathbb{P}(a_i|m_j)$ depends on whether correspondences are in the possible matching $m_i$. In definition 10 we assume that each CCQ is answered independently. Therefore, given that $m_j$ is the correct matching, we know the correct answers for $k$ CCQs and answers $A_{c_{1}}^{(i)}$, $A_{c_{2}}^{(i)}$, ..., $A_{c_{k}}^{(i)}$ are $k$ independent Bernoulli random variables. It follows that\vspace{-1ex}
\begin{footnotesize}
\begin{align*} & \mathbb{P}(a_{i}|m_{j})=\\
 & \quad\prod_{\substack{t=1,...,k\\
c_{t}\in m_{j}\\
A_{c_{t}}^{(i)}=Y
}
}P_{W_{t}}\prod_{\substack{t=1,...,k\\
c_{t}\in m_{j}\\
A_{c_{t}}^{(i)}=N
}
}\left(1-P_{W_{t}}\right)\prod_{\substack{t=1,...,k\\
c_{t}\notin m_{j}\\
A_{c_{t}}^{(i)}=Y
}
}\left(1-P_{W_{t}}\right)\prod_{\substack{t=1,...,k\\
c_{t}\notin m_{j}\\
A_{c_{t}}^{(i)}=N
}
}P_{W_{t}}
\end{align*}
\vspace{-1.5ex}
\end{footnotesize}

\textbf{Running Example:} In the example of Table 1, we assume two CCQs $c_1$ and $c_2$ are answered by two workers with $P_{W_1}=0.8$ and $P_{W_2}=0.6$. Domains of correspondences and answers are\vspace{-1.5ex}
\begin{footnotesize}
\begin{align*}
D_{U} & =\{(T,T),(T,F),(F,T),(F,F)\}\\
D_{A} & =\{(Y,Y),(Y,N),(N,Y),(N,N)\}
\end{align*}
\end{footnotesize}
Probability distribution for $D_U$ is given by Eq~\ref{eq:Pui}:\vspace{-1ex}
\begin{footnotesize}
\begin{align*}
\mathbb{P}(u_{i}=(T,T))=\mathbb{P}(m_{1}) & =0.45\\
\mathbb{P}(u_{i}=(T,F))=\mathbb{P}(m_{2}) & =0.3\\
\mathbb{P}(u_{i}=(F,T))=\mathbb{P}(m_{3}) & =0.25\\
\mathbb{P}(u_{i}=(F,F)) & =0
\end{align*}
\end{footnotesize}Probability distribution for $D_A$ is given by Eq~\ref{eq:2p1}: $\mathbb{P}(a_{i}=(Y,Y))=0.45P_{W_{1}}P_{W_{2}}+0.3P_{W_{1}}(1-P_{W_{2}})+0.25(1-P_{W_{1}})P_{W_{2}}=0.342$.
Similarly, $\mathbb{P}(a_i=(Y,N))=0.308$, $\mathbb{P}(a_i=(N,Y))=0.198$ and $\mathbb{P}(a_i=(N,N))=0.152$. Given $m_1$ is the correct matching, we know that $c_1$ and $c_2$ are T. Thus\vspace{-1ex}
\begin{footnotesize}
\begin{align*}
 & \mathbb{P}(a_{i}=(Y,Y)|m_{1})=P_{W_{1}}P_{W_{2}}=0.48\\
 & \mathbb{P}(a_{i}=(Y,N)|m_{1})=P_{W_{1}}(1-P_{W_{2}})=0.32\\
 & \mathbb{P}(a_{i}=(N,Y)|m_{1})=(1-P_{W_{1}})P_{W_{2}}=0.12\\
 & \mathbb{P}(a_{i}=(N,N)|m_{1})=(1-P_{W_{1}})(1-P_{W_{2}})=0.08
\end{align*}
\vspace{-7ex}
\end{footnotesize}

\subsection{Framework of Multiple CCQ}
\vspace{-1ex}
As shown in Framework~\ref{Multiple CCQ}, the best size-$k$ set of
CCQs are initially selected and published with accuracy rates to show their hardness, and then we constantly
monitor their states. Whenever one or more answers are available,
three operations are conducted. First, all CCQs at state ``waiting''
are withdrawn. Second, the probability distribution is adjusted with
the new answers (line 8\&9). Last, we regenerate and publish a set
of CCQs that are currently most contributive (lines 12\&15). In
general, we keep the best $k$ CCQs in the crowd, by interactively
changing CCQs based on newly received answers. Note that the number
of CCQs may be less than $k$ when the budget is insufficient (line
14-16). The whole procedure terminates when the budget runs out and
all the CCQs are answered (line 3).

In contrast with Single CCQ, the essential query of Multiple CCQ is
to find a group of $k$ CCQs, which  maximize the uncertainty reduction. Formally, we have following definition:\vspace{-1ex}
\begin{definition}[Multiple CCQ Selection (MCCQS)]
\label{MCCQS}
~\\
Given result set \textit{R}, probability assignment function $\mathbb{P}()$, and an integer $k$, the multiple CCQ selection problem is to retrieve a set of $k$ CCQs, denoted by $S_{Q}$, such that the uncertainty reduction, $\Delta H_{S_{Q}}$, is maximized.\vspace{-1ex}
\end{definition}

One can see that, if we set $k=B$ (recall $B$ is the budget of CCQs), the problem of MCCQS selects the optimal set of correspondences at which to ask CCQs in order
to maximize the uncertainty reduction. Similar to \cite{Parameswaran:2011:HGS:1952376.1952377} and \cite{ERCrwodsourcing},  MCCQS itself is an interesting and valuable optimization problem to investigate.\vspace{-2ex}
\begin{salgorithm}
\floatname{algorithm}{\small{Framework}}
\caption{\small{Multiple CCQ}}
\label{Multiple CCQ}
\begin{algorithmic}[1]
 \State $CONS\gets k$ // consumption of the budget
 \State given initial accuracy rates for all correspondences, find and publish a set of CCQs - $S_{Q} = \{Q_{c_1},Q_{c_2},...,Q_{c_k}\}$ that maximize $\Delta H_{S_{Q}}$//(See ~\ref{subsection:Multiple-CCQ-Form})
  \While {there exists CCQs in the crowd, we constantly monitor the CCQs}
 \If{receive the one or more answers $A_1,A_2,..$ with accuracy $P_{W_1},P_{W_2},...$}
  \State withdraw all the CCQs at waiting state
  \State $k'\gets$ $the~number~of~CCQs~withdrawn$
  \State $k''\gets$ $the~number~of~answers~received$
      \For{ each $A_i$,$P_{W_i}$} //Adjustment
          \State $\forall m_i\in R$ Adjust the $\mathbb{P}(m_i)$ by $\mathbb{P}(m_i|A_1,A_2,..)$  //
          \EndFor
\If {$CONS+k''<=B$}
\State find a set of CCQs - $S_{Q}'$ of size $(k'+k'')$ that currently maximize $\Delta H_{S_{Q}'}$ //(See \ref{subsection:Multiple-CCQ-Form})
\State  $CONS = CONS+k''$
\Else  ~// no sufficient budget for maintaining k CCQs
\State find a set of CCQs - $S_{Q}'$ of size $(B-CONS)$ that currently maximize $\Delta H_{S_{Q}'}$ //(See \ref{subsection:Multiple-CCQ-Form})
 \State $CONS = B - k'$
\EndIf
\State $\forall Q_{c_i}'\in S_{Q}'$ publish $Q_{c_i}'$
\EndIf
\EndWhile
\vspace{-1ex}
\end{algorithmic}
\end{salgorithm}
\vspace{-1ex}
\subsection{Simplification of Multiple CCQ Selection}
\label{proofMCCQ}
\vspace{-1ex}
In case of Single CCQ, considering each CCQ as
a random variable, we proved that the uncertainty reduction of a CCQ is equivalent to entropy of answer minus entropy of crowd. In Multiple CCQ, analogously, we are interested to find a relation between uncertainty reduction and entropy for a size-k set of CCQs. This is complex since the correspondences are correlated.

As shown in Theorem~\ref{smp}, under the assumption  that crowds give correct answers with accuracy probability, we prove that the uncertainty
reduction by a set of CCQs is equivalent to their \textit{\textbf{joint entropy}}
(denoted by \textbf{$H(D_A)$}) minus sum of entropies of crowds, while in previous conference paper \cite{DBLP:pvldb/ZhangCJC13}, the result can be viewed as a special case of this result when crowds' accuracies equal to 1. Facilitated with this theorem, we could reduce MCCQS to a special case of joint entropy maximization problem. Similarly with definition 9, the joint entropy $H(D_A)$ of answers $A_{c_{1}}$, $A_{c_{2}}$, ..., $A_{c_{k}}$ w.r.t. CCQs $Q_{c_1}$, $Q_{c_2}$, ..., $Q_{c_k}$ are defined by\vspace{-1ex}
\begin{sequation}\label{HDA}
H(D_A) = -\sum_{a_i\in D_A}\mathbb{P}(a_i)\log\mathbb{P}(a_i)\vspace{-1ex}
\end{sequation}where $\mathbb{P}(a_i)$ can be computed by Eq~\ref{eq:2p1}.
\begin{theorem}
\label{smp}
Given result set \textit{R}, probability assignment function $\mathbb{P}()$, a set of CCQs $S_{Q}=\{Q_{c_1},Q_{c_2},...,Q_{c_k}\}$, answers $A_{c_{1}}$, $A_{c_{2}}$, ..., $A_{c_{k}}$, accuracies of crowd workers $P_{W_1}$, $P_{W_2}$, ..., $P_{W_k}$ in $[0.5,1]$, we have\vspace{-2ex}
\begin{footnotesize}
\[
\Delta H_{S_{Q}} = H(D_A)-\sum_{t=1}^{k}H(W_{t})\vspace{-2ex}
\]
\end{footnotesize}
\end{theorem}
\begin{proof}
By Eq~\ref{eq:multUD_exp}, we have\vspace{-1ex}
\begin{footnotesize}
\begin{align*} & \Delta H_{S_{Q}}\\
 & =H(R)+\sum_{\substack{m_{j}\in R\\
a_{i}\in D_{A}
}
}\mathbb{P}(m_{j})\mathbb{P}(a_{i}|m_{j})\log\mathbb{P}(m_{j})\\
 & \quad+\sum_{\substack{m_{j}\in R\\
a_{i}\in D_{A}
}
}\left[\mathbb{P}(m_{j})\mathbb{P}(a_{i}|m_{j})\log\mathbb{P}(a_{i}|m_{j})-\mathbb{P}(m_{j})\mathbb{P}(a_{i}|m_{j})\log\mathbb{P}(a_{i})\right]\\
 & \coloneqq H(R)+J_{1}+J_{2}+J_{3}\vspace{-1ex}
\end{align*}
\vspace{-1ex}
\end{footnotesize}
By definition 5, we have\vspace{-1ex}
\begin{footnotesize}
\begin{align*}J_{1} & =\sum_{\substack{m_{j}\in R}
}\left[\sum_{\substack{a_{i}\in D_{A}}
}\mathbb{P}(a_{i}|m_{j})\right]\mathbb{P}(m_{j})\log\mathbb{P}(m_{j})\\
 & =\sum_{\substack{m_{j}\in R}
}\mathbb{P}(m_{j})\log\mathbb{P}(m_{j})=-H(R)\vspace{-1ex}
\end{align*}
\end{footnotesize}
By Eq~\ref{HDA}, we have\vspace{-1.3ex}
\begin{footnotesize}
\begin{align*}J_{3} & =-\sum_{\substack{a_{i}\in D_{A}}
}\left[\sum_{\substack{m_{j}\in R}
}\mathbb{P}(m_{j}|a_{i})\right]\mathbb{P}(a_{i})\log\mathbb{P}(a_{i})\\
 & =-\sum_{\substack{a_{i}\in D_{A}}
}\mathbb{P}(a_{i})\log\mathbb{P}(a_{i})=H(D_{A})
\end{align*}
\end{footnotesize}Given $m_i$, $A_{c_{1}}$, ..., $A_{c_{k}}$ are independent. For $J_2$, by the property of joint entropy of independent random variables, we have\vspace{-1ex}
\begin{footnotesize}
\begin{align*}\sum_{\substack{a_{i}\in D_{A}}
}\mathbb{P}(a_{i}|m_{j})\log\mathbb{P}(a_{i}|m_{j}) & =-H\left(\left.A_{c_{1},},...,A_{c_{k}}\right|m_{j}\right)\\
 & =-\sum_{t=1}^{k}H\left(\left.A_{c_{t}}\right|m_{j}\right)
\end{align*}
\end{footnotesize}Therefore, similarly with Eq~\ref{eq:j2j5},\vspace{-1.3ex}
\begin{footnotesize}
\begin{align*}J_{2} & =-\sum_{\substack{m_{j}\in R}
}\mathbb{P}(m_{j})\sum_{t=1}^{k}H\left(\left.A_{c_{t}}\right|m_{j}\right)\\
 & =\sum_{t=1}^{k}\left[\sum_{\substack{m_{j}\in R}
}\mathbb{P}(m_{j})\mathbb{P}(\left.A_{c_{t}}=Y\right|m_{j})\log\mathbb{P}(\left.A_{c_{t}}=Y\right|m_{j})\right.\\
 & \quad+\left.\sum_{\substack{m_{j}\in R}
}\mathbb{P}(m_{j})\mathbb{P}(\left.A_{c_{t}}=N\right|m_{j})\log\mathbb{P}(\left.A_{c_{t}}=N\right|m_{j})\right]\\
 & =\sum_{t=1}^{k}\left[P_{W_{t}}\log P_{W_{t}}+\left(1-P_{W_{t}}\right)\log\left(1-P_{W_{t}}\right)\right]=-\sum_{t=1}^{k}H(W_{t})
\end{align*}
\end{footnotesize}This completes the proof.
\end{proof}\vspace{-4.3ex}

\subsection{Upper bound and lower bound of Uncertainty Reduction}
\vspace{-1ex}
In this subsection, we show the upper and lower bounds for $H(D_A)$, which can be applied to improve approximate algorithm. We recall $U=\{c_{t};t=1,...,k\}$, $D_U$, $\mathcal{P}_U$  Eq~\ref{eq:DUandP} and $\mathbb{P}(u_i)$ Eq~\ref{eq:Pui}. Now we define joint entropy $H(D_{U})$ by\vspace{-1ex}
\begin{sequation}\label{HDU}
H(D_U)=-\sum_{u_{i}\in D_{U}}\mathbb{P}(u_i)\log\mathbb{P}(u_i)\vspace{-1ex}
\end{sequation}We remark that $H(D_U)$ measures the uncertainty of $k$ correspondences, while $H(D_A)$ measures the uncertainty of answers for $k$ correspondences. Intuitively, this difference is caused by the fact that crowds make mistakes. If $P_{W_i}=1$ for all $i=1,...,k$, $H(D_U)=H(D_A)$.

As mentioned in subsection 4.1, the number of elements with positive probability in $D_U$ is at most $|R|$. Time complexity of computing all $\mathbb{P}(u_{i})$ is bounded by $\mathcal{O}(k|R|)$. However $|D_A|=2^k$, by Eq~\ref{eq:2p1}, time complexity of computing all  $\mathbb{P}(a_{i})$  will be $\mathcal{O}(2^k)$. Thus we hope to bound $H(D_A)$ by  $H(D_U)$.\vspace{-0.8ex}
\begin{theorem}
  Under the assumption of Theorem 4.1, let\vspace{-0.8ex}
\begin{footnotesize}
\[
h^{(u)}(D_{A})=\min\left\{ H(D_{U})+\sum_{t=1}^{k}H(W_{t}),-\sum_{t=1}^{k}\log\left(1-P_{W_{t}}\right)\right\}\vspace{-1ex}
\]
\end{footnotesize}
  and\vspace{-0.8ex}
\begin{footnotesize}
\begin{align*}
 & h^{(l)}(D_{A})\\
 & =\max\left\{ -\sum_{t=1}^{k}\log P_{W_{t}},\quad H(D_{U})+\sum_{t=1}^{k}H(W_{t})\right.\\
 & \quad+\prod_{\substack{t=1}
}^{k}P_{W_{t}}\log\prod_{\substack{t=1}
}^{k}P_{W_{t}}+\left(1-\prod_{\substack{t=1}
}^{k}P_{W_{t}}\right)\log\left(1-\prod_{\substack{t=1}
}^{k}P_{W_{t}}\right)\\
 & \quad\left.-\left(1-\prod_{\substack{t=1}
}^{k}P_{W_{t}}\right)\min\left\{ \log(2^{k}-1),H(D_{U})\right\} \right\}
\end{align*}\end{footnotesize}We have
\begin{sequation}\label{UandLbound}
 h^{(l)}(D_A)\leq H(D_A)\leq h^{(u)}(D_A)\vspace{-1ex}
\end{sequation}
\end{theorem}

\begin{proof}
\textbf{Upper bound}: By the chain rule of conditional entropy, we have\vspace{-1ex}
\begin{sequation}\label{ChainRuleForUB}
\begin{aligned}
H(D_{A}) & =H(D_{A},D_{U})-H(D_{U}|D_{A})\\
 & =H(D_{A}|D_{U})+H(D_{U})-H(D_{U}|D_{A})\\
\end{aligned}\vspace{-1ex}
\end{sequation}where\vspace{-1ex}
\begin{footnotesize}
\[
H(D_{A}|D_{U})=-\sum_{\substack{u_{j}\in D_{U}}
}\mathbb{P}(u_{j})\left[\sum_{\substack{a_{i}\in D_{A}}
}\mathbb{P}(a_{i}|u_{j})\log\mathbb{P}(a_{i}|u_{j})\right]\vspace{-1ex}
\]
\end{footnotesize}Given $u_j$, we know the true correspondences and false ones in $U$, thus $A_{c_t}$, $t=1,...,k$ are independent. We obtain that\vspace{-1ex}
\begin{footnotesize}
\begin{align*} & H(D_{A}|D_{U})\\
 & =\sum_{\substack{u_{j}\in D_{U}}
}\mathbb{P}(u_{j})\sum_{t=1}^{k}H\left(\left.A_{c_{t}}\right|u_{j}\right)\\
 & =-\sum_{t=1}^{k}[\sum_{\substack{u_{j}\in D_{U}}
}\mathbb{P}(u_{j})\mathbb{P}(\left.A_{c_{t}}=Y\right|u_{j})\log\mathbb{P}(\left.A_{c_{t}}=Y\right|u_{j})\\
 & \quad+\sum_{\substack{u_{j}\in D_{U}}
}\mathbb{P}(u_{j})\mathbb{P}(\left.A_{c_{t}}=N\right|u_{j})\log\mathbb{P}(\left.A_{c_{t}}=N\right|u_{j})]\\
 & =-\sum_{t=1}^{k}[\sum_{\substack{u_{j}\in D_{U}\\
c_{t}^{(j)}=F
}
}\mathbb{P}(u_{i})\left(1-P_{W_{t}}\right)\log\left(1-P_{W_{t}}\right)\\
 & \quad+\sum_{\substack{u_{j}\in D_{U}\\
c_{t}^{(j)}=T
}
}\mathbb{P}(u_{i})P_{W_{t}}\log P_{W_{t}}+\sum_{\substack{u_{j}\in D_{U}\\
c_{t}^{(j)}=F
}
}\mathbb{P}(u_{i})P_{W_{t}}\log P_{W_{t}}\\
 & \quad+\sum_{\substack{u_{j}\in D_{U}\\
c_{t}^{(j)}=T
}
}\mathbb{P}(u_{i})\left(1-P_{W_{t}}\right)\log\left(1-P_{W_{t}}\right)]\\
 & =\sum_{t=1}^{k}H(W_{t})
\end{align*}
\end{footnotesize}Thus we get\vspace{-1ex}
\begin{sequation}
H(D_A)\leq H(D_U)+\sum_{t=1}^{k}H(W_{k})\label{eq:up1}\vspace{-1.5ex}
\end{sequation}On the other hand, by definition of $H(D_A)$ Eq~\ref{HDA} and $\mathbb{P}(a_{i})$ Eq~\ref{eq:2p1}, we have\vspace{-1ex}
\begin{sequation}
H(D_{A})=-\sum_{i=1}^{2^{k}}\sum_{j=1}^{2^{k}}\mathbb{P}(u_{j})q_{ij}\log\mathbb{P}(u_{j})q_{ij}\label{eq:up2}\vspace{-1ex}
\end{sequation}Note that $\sum_{j=1}^{2^{k}}\mathbb{P}(u_{j})=1$, which means $\sum_{j=1}^{2^{k}}\mathbb{P}(u_{j})q_{ij}$ is a linear combination of $q_{ij}$, $j=1,2,...,2^k$. It is easy to see that\vspace{-1ex}
\begin{footnotesize}
\[
\sum_{j=1}^{2^{k}}\mathbb{P}(u_{j})q_{ij}\geq\min_{j}q_{ij}=\prod_{t=1}^{k}\left(1-P_{W_{t}}\right)\vspace{-1ex}
\]
\end{footnotesize}where last equation holds because each $P_{W_t} \in [0.5,1]$. Then we have\vspace{-1ex}
\begin{footnotesize}
\[
H(D_{A})\leq-\sum_{i=1}^{2^{k}}\sum_{j=1}^{2^{k}}\mathbb{P}(u_{j})q_{ij}\log\prod_{t=1}^{k}\left(1-P_{W_{t}}\right)=-\sum_{t=1}^{k}\log\left(1-P_{W_{t}}\right)\vspace{-1ex}
\]
\end{footnotesize}Together with Eq~\ref{eq:up1}, we achieve the upper bound.

\textbf{Lower bound}: The difference between $H(D_A)$ and $H(D_U)$ is that crowds have probability to make mistakes. Inspired by this, we consider the indicator function that crowds make at least one mistake, i.e.\vspace{-1ex}
\begin{sequation}\label{indicator}
\begin{aligned}
Y=\begin{cases}
0 & \mathrm{Answers\;are\;all\;correct}\\
1 & \mathrm{Crowds\;make\;mistake}
\end{cases}
\vspace{-2ex}
\end{aligned}
\vspace{-2ex}
\end{sequation}Obviously, we have $\mathbb{P}(Y=0)=\prod_{\substack{t=1}}^{k}P_{W_{t}}$.
In order to obtain lower bound, it is sufficient to bound the term $H(D_U|D_A)$ in Eq~\ref{ChainRuleForUB}. Thus we rewrite\vspace{-2ex}
\begin{footnotesize}
\begin{align}
 & H\left(D_{U}\left|D_{A}\right.\right)\nonumber \\
 & \quad=H\left(D_{U}\left|D_{A}\right.\right)-H\left(D_{U}\left|D_{A},Y\right.\right)+H\left(D_{U}\left|D_{A},Y\right.\right)\nonumber \\
 & \quad=H\left(Y\left|D_{A}\right.\right)-H\left(Y\left|D_{A},D_{U}\right.\right)+H\left(D_{U}\left|D_{A},Y\right.\right)\label{eq:chainrule2}\\
 & \quad=H\left(Y\left|D_{A}\right.\right)+H\left(D_{U}\left|D_{A},Y\right.\right)\nonumber \\
 & \quad\leq H\left(Y\right)+H\left(D_{U}\left|D_{A},Y\right.\right)\nonumber \\
 & \quad=H\left(Y\right)+\sum_{\substack{a_{i}\in D_{A}}
}\left[\mathbb{P}(a_{i},Y=0)H\left(D_{U}\left|a_{i},Y=0\right.\right)\right.\label{eq:lb1}\\
 & \qquad+\left.\mathbb{P}(a_{i},Y=1)H\left(D_{U}\left|a_{i},Y=1\right.\right)\right]\label{eq:lb2}\
\end{align}\end{footnotesize}where the second equation Eq~\ref{eq:chainrule2} is obtained by chain rule of entropy. Please note that when $Y=0$, $A_{c_t}^{(i)}=Y$ if $c_{t}^{(i)}=T$ and $A_{c_t}^{(i)}=N$ if $c_{t}^{(i)}=F$. Thus in Eq~\ref{eq:lb1}, we have\vspace{-1.5ex}
\begin{footnotesize}
\[
H\left(D_{U}\left|a_{i},Y=0\right.\right)=0\vspace{-1ex}
\]\end{footnotesize}The entropy is maximized when each possible outcome has the same probability. Since $|D_U|=2^k$ and when $Y=1$, we know that the number of possible outcome is $2^k-1$. Therefore in  Eq~\ref{eq:lb2}, we have\vspace{-2ex}
\begin{footnotesize}
\[
H\left(D_{U}\left|a_{i},Y=1\right.\right)\leq\min\left\{ H(D_{U}),\log\left(2^{k}-1\right)\right\}\vspace{-1ex}
\]\end{footnotesize}Now we write\vspace{-1.5ex}
\begin{footnotesize}
\begin{align*} & H\left(D_{U}\left|D_{A}\right.\right)\\
 & \quad\leq H\left(Y\right)+\min\left\{ H(D_{U}),\log\left(2^{k}-1\right)\right\} \sum_{\substack{a_{i}\in D_{A}}
}\mathbb{P}(a_{i},Y=1)\\
 & \quad=H\left(Y\right)+\min\left\{ H(D_{U}),\log\left(2^{k}-1\right)\right\} \mathbb{P}(Y=1)\\
 & \quad=-\prod_{\substack{t=1}
}^{k}P_{W_{t}}\log\prod_{\substack{t=1}
}^{k}P_{W_{t}}-\left(1-\prod_{\substack{t=1}
}^{k}P_{W_{t}}\right)\log\left(1-\prod_{\substack{t=1}
}^{k}P_{W_{t}}\right)\\
 & \qquad+\left(1-\prod_{\substack{t=1}
}^{k}P_{W_{t}}\right)\min\left\{ H(D_{U}),\log\left(2^{k}-1\right)\right\}
\end{align*}
\end{footnotesize}Substitute this bound into Eq~\ref{ChainRuleForUB}, we achieve that\vspace{-2ex}
\begin{footnotesize}
\begin{align*}
H\left(D_{A}\right) & \geq H(D_{U})+\sum_{t=1}^{k}H(W_{t})+\prod_{\substack{t=1}
}^{k}P_{W_{t}}\log\prod_{\substack{t=1}
}^{k}P_{W_{t}}\\
 & \quad+\left(1-\prod_{\substack{t=1}
}^{k}P_{W_{t}}\right)\log\left(1-\prod_{\substack{t=1}
}^{k}P_{W_{t}}\right)\\
 & \quad-\left(1-\prod_{\substack{t=1}
}^{k}P_{W_{t}}\right)\min\left\{ H(D_{U}),\log\left(2^{k}-1\right)\right\}
\end{align*}
\end{footnotesize}
On the other hand by Eq~\ref{eq:up2}, we have\vspace{-2ex}
\begin{footnotesize}
\[
H(D_{A})\geq-\sum_{i=1}^{2^{k}}\sum_{j=1}^{2^{k}}\mathbb{P}(u_{j})q_{ij}\log\prod_{t=1}^{k}P_{W_{t}} =-\sum_{t=1}^{k}\log P_{W_{t}}\vspace{-1ex}
\]
\end{footnotesize}This completes the proof.
\end{proof}
\vspace{-1ex}
\textbf{Remark:} When $P_{W_t}=1$ for all $t=1,...,k$, we can check that\vspace{-1ex}
\begin{footnotesize}
\[
h^{(l)}(D_{A})=h^{(u)}(D_{A})=H(D_U)=H(D_A)\vspace{-1ex}
\]\end{footnotesize}When $P_{W_t}=0.5$ for all $t=1,...,k$, we can check that\vspace{-1ex}
\begin{footnotesize}
\[
h^{(l)}(D_{A})=h^{(u)}(D_{A})=k=H(D_A)\vspace{-1ex}
\]\end{footnotesize}Our result is optimal in the sense that lower bound equals to upper bound in two extreme cases: When crowds always give correct answers ($P_{W_t}=1$) and when crowds always give random answers without any consideration ($P_ {W_t}=0.5$).\vspace{-2.5ex}

\subsection{NP-hardness of Multiple CCQ Selection}
\vspace{-1.2ex}
By Theorem~\ref{smp}, searching a group of $k$ CCQs with maximal uncertainty reduction is equivalent to \textit{finding k CCQs with maximal joint entropy}. It is known the joint entropy of a set of random variables is a \textit{monotone sub-modular function}. In general, maximizing sub-modular functions is NP-hard. Concerning the computation of the value of information, \cite{DBLP:Submodular} shows that, for a general reward function $R_j$ (in our problem, $R_j = \Delta H_{S_{Q}}$), it is $NP^{PP}-hard$ to select the optimal subset of variables even for discrete distributions that can be represented by polytree graphical models. $NP^{PP}-hard$ problems are believed to be much harder than $NPC$ or $\#PC$ problems. In the problem of multiple CCQ selection, every variable is binary and their marginal distribution is represented by a binary matrix. As a result, a naive traversal would lead to an algorithm of $\mathcal{O}(|R||C|^k)$ complexity, since the searching space (i.e. the number of subsets to select) is always of size $C_{|C|}^k$.

With the Theorem~\ref{nphard}, we prove that Multiple CCQ Selection is NP-hard. Encountering this NP-hardness, we propose a efficient approximation algorithm based on the sub-modularity of joint entropy.
\begin{theorem}\vspace{-1.9ex}
\label{nphard}
The Multiple CCQ Selection is NP-hard.\vspace{-1.9ex}
\end{theorem}
\begin{proof}
To reach the proof of Theorem~\ref{nphard}, it is sufficient to prove the NP-completeness of its decision version, Decision MCCQS (DMCCQS), i.e. given result set \textit{R}, probability assignment function $\mathbb{P}()$, an integer k, and a value $\Delta H$, decide whether one can find a set $S_{Q}$ of k CCQs such that $\Delta H_{S_{Q}}>=\Delta H$.

To reach the NP-completeness of DMCCQS, it is sufficient to prove a special case of DMCCQS is NPC. First we let accuracy rates equals to 1. Moreover we state the special case of DMCCQS by adding the following constraint on $R$: for each way of partitioning $R$ into two subsets $S_1$ and $S_2$, there exists a correspondence $c$ such that
$(\forall m_i\in S_1, c\in m_i) \wedge (\forall m_j\in S_2, c\notin m_i)$. Equipped with this constraint, we this reduce special case of DMCCQS to the \textit{set partition problem}.

The partition problem is the task of deciding whether a given multiset of positive integers can be partitioned into two subsets $ S_1 $ and $ S_2 $ such that the sum of the numbers in $ S_1 $ equals the sum of the numbers in $ S_2 $.

\textbf{Transformation:}
Given a set partition problem with input multiset $S$, let $Sum = \sum_{x\in S}x$. We create a possible matching $m_i$ for each positive integer $x_i\in S$, and assign its possibility $\mathbb{P}(m_i)= x_i/Sum $. Let the correspondences satisfy the constraint, and we set $k=1, \Delta H= -\log(0.5)=1 $ for DMCCQS.

\textbf{($\Longrightarrow$)}
If there is a yes-certificate for the set partition problem, then the $R$ can be partitioned into two subsets, each with aggregate probability 0.5. According to the constraint, the exists a correspondence c with $\mathbb{P}(c)=0.5$. Then, selecting $S_{Q}=\{Q_c\}$ would achieve uncertainty reduction $H_{S_{Q}} = -0.5\log{0.5}-0.5\log{0.5}=1$. Therefore,$\{Q_c\}$ serves as yes-certificate for the special case of DMCCQS.

\textbf{($\Longleftarrow$)}
Assume there is yes-certificate for the special case of DMCCQS when $k=1, \Delta H= -\log(0.5)=1 $. Since $k=1$, $H_{S_{Q}} $is actually equivalent to $H_c$. Then by Theorem~\ref{Simp-Single-c}, there exists a correspondence $c$ such that $\mathbb{P}(c)=0.5$. Therefore, by the constraint, there is a way to partition $R$ into two subsets, each with aggregate probability 0.5. Since the mapping from the positive integers to the possible matchings is one-to-one, we obtain an yes-certificate for the special case of DMCCQS.\vspace{-1.9ex}
\end{proof}

\subsection{Approximation Algorithm}
\vspace{-1.2ex}
It is known that the joint entropy of a set of random variables is a \textit{monotone sub-modular function} \cite{DBLP:Submodular}. And the problem of selecting a k-element subset maximizing a monotone sub-modular function can be approximated with a performance guarantee of $(1-1/e)$, by iteratively selecting the most uncertain variable given the ones selected so far \cite{DBLP:journals/ipl/KhullerMN99}. Formally, we have the optimization function at the $k^{th}$ iteration:\vspace{-1ex}
\begin{sequation}
\begin{aligned}
\label{optimizationGoal}
X:= \arg\max_{Q_{c_k}} \Delta H_{S_Q^{k-1}\cup\{Q_{c_k}\}}
\end{aligned}
\end{sequation}
Let $A^{(k-1)}$ denote answers for $S_Q^{k-1}$. By the chain rule of conditional entropy, we have\vspace{-1ex}
\begin{footnotesize}
\[
H\left(D_{A^{(k-1)}},A_{c_k}\right)=H\left(D_{A^{(k-1)}}\right)+H\left(A_{c_{k}}\left|D_{A^{(k-1)}}\right.\right)\vspace{-1ex}
\]\end{footnotesize}Thus we only need to maximize the conditional entropy at each iteration, i.e.\vspace{-1ex}
\begin{footnotesize}
\[
X:= \arg\max_{A_{c_k}} H\left(A_{c_k}\left|D_{A^{(k-1)}}\right.\right)\vspace{-1ex}
\]\end{footnotesize}and\vspace{-1ex}
\begin{sequation}
\begin{aligned} & H\left(A_{c_{k}}\left|D_{A^{(k-1)}}\right.\right)\\
 & \quad=-\sum_{a_{i}\in D_{A^{(k-1)}}}\mathbb{P}(a_{i})\left[\mathbb{P}\left(\left.A_{c_{k}}=Y\right|a_{i}\right)\log\mathbb{P}\left(\left.A_{c_{k}}=Y\right|a_{i}\right)\right.\\
 & \qquad\left.+\mathbb{P}\left(\left.A_{c_{k}}=N\right|a_{i}\right)\log\mathbb{P}\left(\left.A_{c_{k}}=N\right|a_{i}\right)\right]
\end{aligned}
\vspace{-3ex}
\label{eq:opt}
\end{sequation}

\begin{figure}[!t]
  \centerline{\psfig{figure=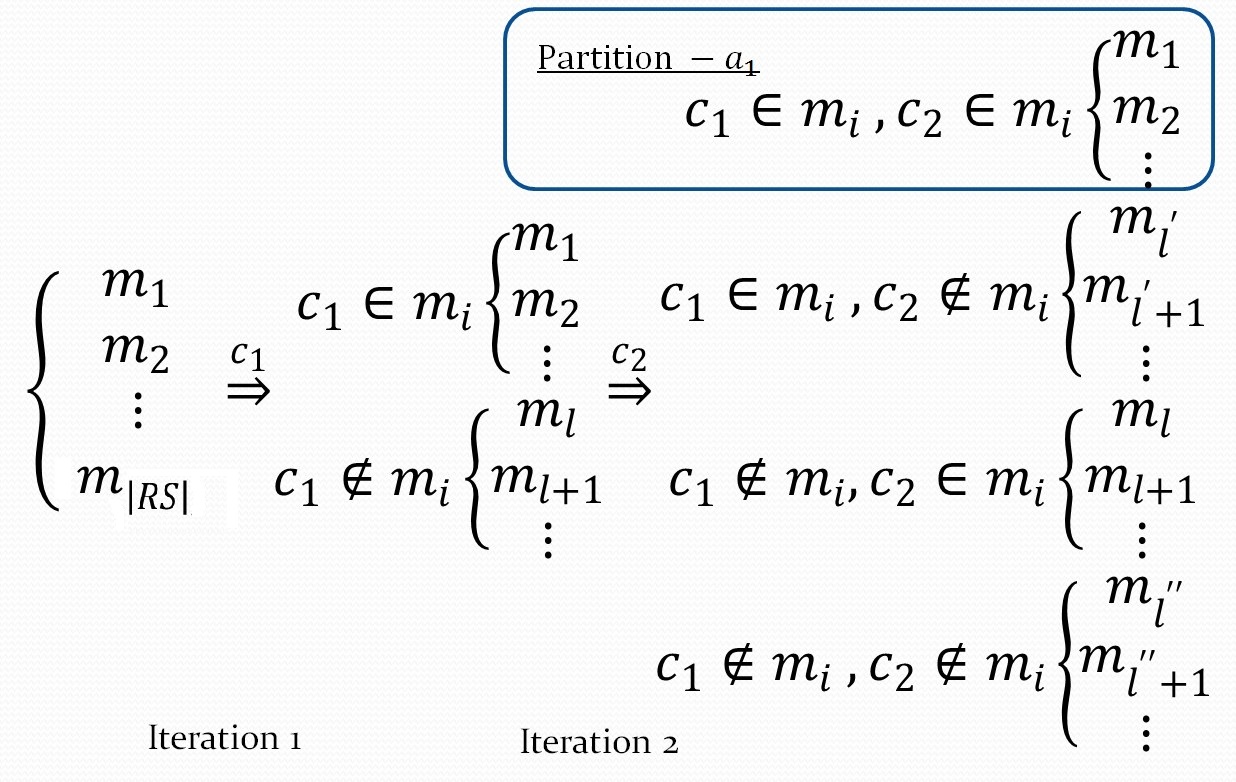,height=33mm} }\vspace{-3ex}
    \caption{\footnotesize{Illustration of R Partitioning}}\vspace{-4.5ex}
   \label{fig:partition}
\end{figure}

Eq~\ref{eq:opt} indicates that, at each iteration, we are searching the
most uncertain correspondence, given the correspondences selected in
previous iterations. In particular, after the $(k-1)^{th}$
iteration, the possible matchings are at most split into $2^{k-1}$
partitions, each of which corresponds to an element $a_i\in
D_{A^{(k-1)}}$. We aim to find the $k^{th}$ correspondence, in
order to further split them to at most $2^k$ partitions, such that
then entropy of resulting partitions is maximized.
Figure~\ref{fig:partition} illustrates a partitioning of the first two
iterations. Motivated with this interpretation, we propose to apply
an in-memory index to maintain the list of partitions for each
iteration. One can see that each partition corresponding to $a_i$ is
essentially a set of possible matchings. In addition, also index
$\mathbb{P}(a_i)$ associated with each partition.

As a result, the computation of $H(A_{c_k}|D_{A^{(k-1)}})$ for each
candidate correspondence is simply traversing the list of
partitions. Note the number of partitions is at most $|R|$ (i.e.
each partition has only one possible matching), so the overall
complexity is upper bounded by $\mathcal{O}(k|R||C|)$. However, there is
still room for the further pruning of the search space. In the
follows, we derive four pruning techniques to avoid traversing all
the partitions. Each pruning indicates a condition that guarantees
certain partitions are unnecessary to be considered, hence speed up
the overall computation. For simplicity, we just use the notation
$a_i$ to represent the partition corresponding to $a_i$. Then, for
the iteration, we have partitions ${a_1,a_2,...,a_n}$ with
probabilities ${\mathbb{P}(a_1),\mathbb{P}(a_2),...,\mathbb{P}(a_n)}$ respectively. As
follows, we present four pruning rules.\vspace{-1ex}
\begin{pruning rule}
\label{pruning_1} If a partition $a_i$ has only one matching, $a_i$
can be safely pruned, i.e. we can remove $a_i$ from the list of
partitions.\vspace{-1ex}
\end{pruning rule}
Pruning rule~\ref{pruning_1} utilizes the intuition that the
correctness of a possible matching $m$ can be fully determined by
the selected correspondences, when $m$ is the only one in its
partition. In other words, the remaining correspondences of $m$
would not contribute any more information, hence should not be
selected.\vspace{-1ex}
\begin{pruning rule}
\label{pruning_2}
Let $c$ be a candidate correspondence, then $c$ can be safely pruned (for the rest of the iterations), if all $a_i$, one of the following conditions are met for :(1) $\forall m_i\in a_i, c\in m_i$, (2)$\forall m_i\in a_i, c\notin m_i$\vspace{-1ex}
\end{pruning rule}

Similar to Pruning rule~\ref{pruning_1}, Pruning
rule~\ref{pruning_2} indicates the condition that the correctness of
$c$ can be determined by selected correspondences.

Next, we introduce Pruning Rule \ref{pruning_3} and
\ref{pruning_4}, which derives two non-trivial upper bounds, which
enable effective pruning.\vspace{-0.5ex}
\begin{pruning rule}
\label{pruning_3}
Let $H_{best\_so\_far}$ be the best value of \\Eq~\ref{eq:opt} so far for the current iteration, then for the correspondence $c$, let ${a_1,a_2,...,a_m}$ be the partitions $c$ already traversed Let\vspace{-1ex}
\begin{footnotesize}
\begin{align*}
H_{0}=-\sum\limits _{i=1}^{m}\mathbb{P}(a_{i})\left[\mathbb{P}(Y|a_{i})\log\mathbb{P}(Y|a_{i})\right.+\left.\mathbb{P}(N|a_{i})\log\mathbb{P}(N|a_{i})\right]
\vspace{-1ex}
\end{align*}
\end{footnotesize}Then $c$ can be pruned for the current iteration, if we have\vspace{-1ex}
\begin{footnotesize}
\[
H_{best\_so\_far} - H_0 \leq \sum\limits_{j=m+1}^{n}\mathbb{P}(a_j)\log\dfrac{\mathbb{P}(a_j)}{2}\vspace{-1ex}
\]\end{footnotesize}
\end{pruning rule}
\begin{proof}
For the rest of partitions ${a_{m+1},a_{m+2},...,a_n}$, the optimal situation is they are all perfectly bisected, that is $\forall i\in [m+1,n]$, $\mathbb{P}(Y|a_i) =\mathbb{P}(N|a_i)=0.5$. Therefore, their contribution to the optimization function has a upper bound\vspace{-1ex}
\begin{footnotesize}
\[
\sum\limits_{j=m+1}^{n}\mathbb{P}(a_j)\log\dfrac{\mathbb{P}(a_j)}{2}\vspace{-1.5ex}
\]\end{footnotesize}
\end{proof}
\begin{pruning rule}
\label{pruning_4}
Let $H_{best\_so\_far}$ be the best value of \\Eq~\ref{eq:opt} so far for the current iteration. For a correspondence $c$, let \\$H(A_{c}|D_{A^{(k-2)}})$ be the conditional entropy computed from a previous iteration. Then, $c$ can be pruned for the current iteration if\vspace{-1ex}
\begin{footnotesize}
\[
H\left(A_{c}\left|D_{A^{(k-2)}}\right.\right)\leq H_{best\_so\_far}\vspace{-1ex}
\]\end{footnotesize}
\end{pruning rule}
\begin{proof}
This pruning rule reflects the sub-modularity of the joint entropy. $S_Q^{k-2}$ is the set of CCQs selected in the previous iteration, so  $S_Q^{k-2}\subset S_Q^{k-1}$, where $S_Q^{k-1}$ is the CCQs selected for the current iteration. Then by sub-modularity, we have\vspace{-1ex}
\begin{footnotesize}
\[
H\left(D_{A^{(k-2)}},A_{c}\right)-H\left(D_{A^{(k-2)}}\right)\geq H\left(D_{A^{(k-1)}},A_{c}\right)-H\left(D_{A^{(k-1)}}\right)\vspace{-1ex}
\]\end{footnotesize}and equivalently, $H\left(A_{c_k}\left|D_{A^{(k-2)}}\right.\right)\geq H\left(A_{c_k}\left|D_{A^{(k-1)}}\right.\right)$, which completes the proof.\vspace{-1ex}
\end{proof}At last we use Theorem 4.2 to show a pruning rule.\vspace{-1ex}
\begin{pruning rule}
\label{pruning_5}
Given the selected correspondences $S_{Q}^{k-1}$ in previous (k-1)th iterations, two current potential selected correspondences $c_1$ and $c_2$, correspondence $c_2$ could be safely filtered if these two correspondences satisfy\vspace{-1ex}
\begin{footnotesize}
\begin{align*}
h^{(l)}(D_{A^{(k-1)}} \cup A_{c_1}) \geqslant h^{(u)}(D_{A^{(k-1)}} \cup A_{c_2})
\end{align*}\end{footnotesize}
\end{pruning rule}\vspace{-3.5ex}

\section{Experimental Results}
\vspace{-1ex}
\label{Experiment} We conducted extensive experiments to evaluate
our approaches, based on both simulation and real implementation. We
focus on evaluating two issues. First, we examine the effectiveness
of our two frameworks in reducing the uncertainty for possible
matchings. Second, we verify the correctness of our approaches, by
evaluating the precision and recall of the best matchings.
\vspace{-2.8ex}
\subsection{Experimental Setup}
\vspace{-1ex}
We adopt the schema matching tool OntoBuilder
\cite{DBLP:journals/vldb/GalATM05,DBLP:journals/jods/Gal06}, which
is one of the leading tools for schema matching. In particular, we
conduct our experiments on five datasets, each of which includes
five schemata. The schemata are extracted from web forms from
different domains. We describe the characteristics of each dataset
in Table~\ref{datasets}. By OntoBuilder, schemata are parsed
into xml schemata, and attributes refer to nodes with semantic
information. We conduct pairwise schema matching within each domain,
so there are totally 40 pairs of schemata (10 for each domain). In
OntoBuilder, four schema matching algorithms are implemented, namely
\textit{Term}, \textit{Value}, \textit{Composition} and
\textit{Precedence}. For each pair of schemata, we generate 400
unique possible matchings (100 for each algorithm). In addition,
each possible matching is associated with a global score, which
indicates the goodness of the matching. We obtain the probabilities
of matchings by normalizing the global scores. The details of these
algorithms can be found in \cite{DBLP:journals/jods/Gal06}.\vspace{-2.5ex}
\subsection{Simulation}
\vspace{-1ex}
\label{simulation}
\begin{figure}[t]
\vspace{-2ex}
\includegraphics[width=8.5cm]
 {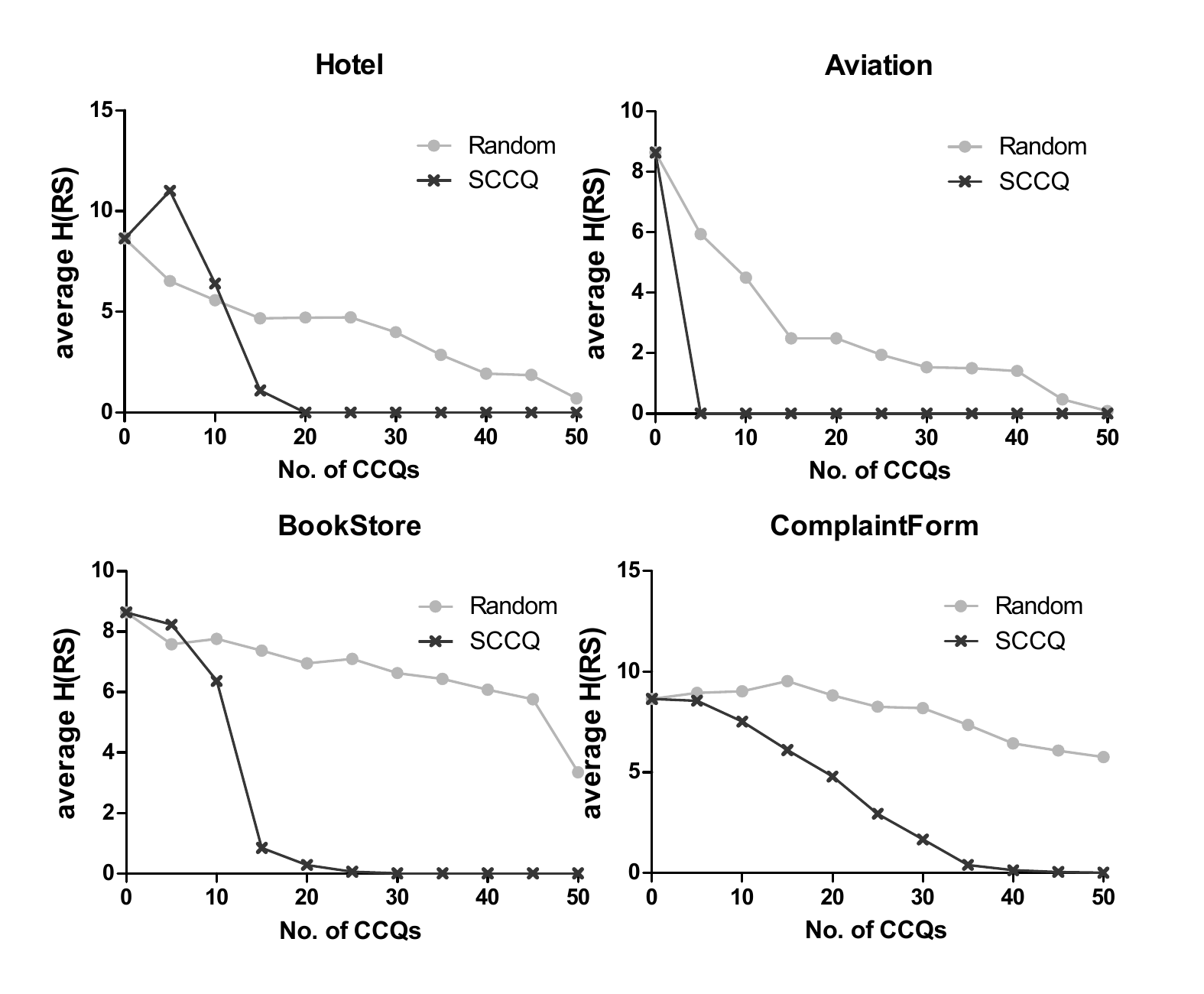}\vspace{-4.5ex}
    \caption{\footnotesize{Single CCQ v.s. Random - Simulation}}\vspace{-4.5ex}
    \label{fig:SCCQ}
\end{figure}
To evaluate the effectiveness of our two approaches, we first
conduct a simulation of the crowd's behaviour, based on our
formulation in Section~\ref{section:formalization}. First, we
manually select the best matching from the 400 possible matchings,
and treat the selected matching as the correct matching (i.e. ground
truth). So for any correspondence, its correctness depends on
whether it is in the selected matching. Second, for each published
CCQ, we randomly generate an accuracy rate $P_W\in [0.5,1]$ following
an uniform distribution. Third, given a CCQ, we generate the correct
yes-no answer with probability $P_W$ (i.e. generate the wrong
answer with probability $1-P_W$, and then return the answer and $P_W$
as the inputs for adjustment (Section 3.1).

First, we present the effectiveness of Single CCQ approach
( Framework~\ref{SingleCCQ}), by comparing its performance with
randomly selecting CCQs. We set the budget $B=50$, and each CCQ is
generated after receiving the answer of the previous one.
Figure~\ref{fig:SCCQ} illustrates the average change of uncertainty
(vertical axis) with the number of answers of CCQs received
(horizontal axis). With the increase of number of CCQs, the
uncertainty converges to zero rapidly. From the experimental
results, our proposed Single CCQ approach (SCCQ) outperforms the
random approach (Random) significantly. Please note that all the results plotted in Section~\ref{simulation} and \ref{AMT} are averages over 10 runs. The distribution is quite dense within each domain, but diverse for different domains.
\nop{GOOD TO SAY HERE THAT ALL RESULTS PLOTTED ARE AVERAGES OVER 40?? RUNS.
AND ONE SENTENCE ABOUT WHAT THE DISTRIBUTION LOOKS LIKE AROUND THE AVERAGE.}

Next, we examine the performance of Multiple CCQ
(Framework~\ref{Multiple CCQ}). Recall that we need to constantly
monitor the CCQs, and update the CCQs whenever new answers are
received. In the simulation of conference paper \cite{DBLP:pvldb/ZhangCJC13}, we check the states of published CCQs every time unit. Each published CCQ is initially at state ``waiting".
For each time unit, each CCQ in state ``waiting" may change to
``accepted'' with probability $P_0$ (remain unchanged with
probability $1-P_0$), where $P_0$ is a random number generated from
$(0,0.5)$; and each CCQ at state ``accepted'' may change to
``answered'' with probability $P_1$ (remain unchanged with
probability $1-P_1$), where $P_1$ follows a Poisson distribution. Figure~\ref{fig:MCCQ}
illustrates the performance of Multiple CCQ by varying k, where we
set the budget $B=50$. Recall that k, a parameter of
Framework~\ref{Multiple CCQ}, represents the number of CCQ in the
crowd. Whenever a CCQ is answered, we dynamically updated the k
CCQs, to make sure the k CCQs are the best according to the all
received answers. In particular, when k=1, Framework~\ref{Multiple
CCQ} becomes the Single CCQ approach. One can observe that the
curves with smaller k tend to have better performance in terms of
reducing uncertainty. In fact, the larger k is, the less advantage MCCQ has comparing to a random selection. Recall each time we select k out of $|C|$ correspondences, and when $k=|C|$, MCCQ is the same as random selection, i.e. select all of the correspondences we have.

As discussed in Section~\ref{MCCQA}, the increase of $k$ leads to less uncertainty reduction (which is consistent with the result in Figure~\ref{fig:MCCQ}), but improves the overall time efficiency. Since there are
multiple uncontrollable factors affecting the completion time of
workers, the time cost of the proposed approaches are hard to be
simulated. Nevertheless, we analyse the relation between $k$ and the time cost in the real-world implementation in Section~\ref{AMT}.\vspace{-2.5ex}
\begin{table}[!t]
\caption{\footnotesize{DATASETS}}\vspace{-2.5ex} 
\label{datasets}
\centering 
\footnotesize
\begin{tabular}{c c c } 

\hline 
\small{Notation} & \small{Source} & \small{No.of attributes}\\ 
Hotel & hotel searching websites & 14-20  \\
Aviation & homepages of airline companies & 12-18\\
BookStore & the webpages of advanced \\
&search in online book stores & 13-21\\
ComplaintForm & the complaint forms of \\
& government websites & 27-34\\
News & news websites & 43-60\\

\end{tabular} 
\vspace{-3.5ex}
\end{table}

\begin{figure}[!t]
    \includegraphics[width=8cm]{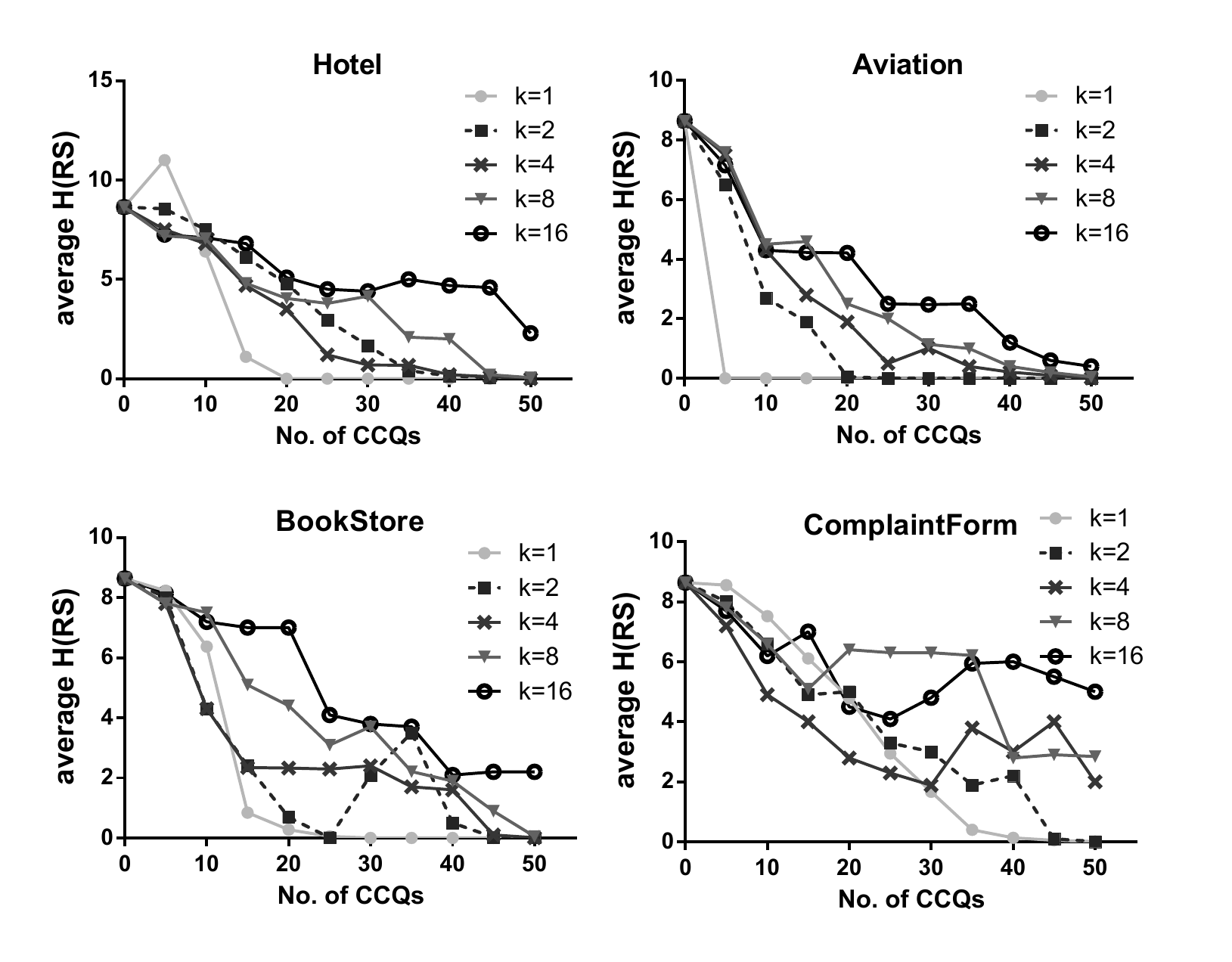}
       \vspace{-5ex}
       \caption{\footnotesize{Multiple CCQ with different k - Simulation}}
       \vspace{-4ex}
    \label{fig:MCCQ}
\end{figure}
\vspace{-1ex}
\subsection{Testing on Amazon Mechanical Turk}
\vspace{-1ex}
\label{AMT} We implement our two approaches on Amazon Mechanical
Turk (AMT), which is a widely used crowdsourcing marketplace.
Empowered with the Amazon Mechanical Turk SDK, we are able
to interactively publish and manage the CCQs. Each HIT of AMT
includes all the attributes of two schemata, one CCQ, and the URLs
of the source web-pages. Each HIT is priced US\$0.05. One can see
that each HIT is essentially a CCQ. For the rest of this section,
the terms ``HIT'' and ``CCQ'' are exchangeable.

\nop{In addition, the
states of published CCQs are checked every 5 seconds. In fact, we would like to check the states as frequently as possible. However, too frequent access may lead to unstable networking connection. The connection became stable when this waiting time is set to at least 5 seconds.}

\nop{Concerning the error rate of the workers in AMT, a qualification
test including 10 CCQs is required for each worker. Then the  error
rate is computed as the fraction of CCQs answered incorrectly in the
qualification test. }

\nop{
Please note that workers are not paid for taking the qualification test, and workers with error rate lower than $40\%$ are eligible for accepting our HITs. In addition, as a fact, workers who correctly
answered all the CCQs in the qualification test may still provide
incorrect answers for further CCQs. Therefore, in case all the
CCQs are answered correctly, the error rate of the worker is set to $0.01$. Besides, for each schema matching problem, a worker is only allowed
to answer one CCQ. This is because we assume the crowdsourced
answers are independent for each schema matching problem. Of course,
workers are allowed answer CCQs generated for different schema
matching problems. Throughout the whole experiment in AMT, there are totally 412 distinct workers, with average error rate $28\%$.}

In analogy to the simulation, Figure~\ref{fig:SCCQReal} and
Figure~\ref{fig:MCCQReal} illustrate the performances of Single CCQ
and Multiple CCQ respectively, where we set the budget $B=50$. In
terms of uncertainty reduction, one can see that the performance
is basically consistent with the simulation. A very important
finding is that, in contrast with the simulation, the
uncertainty is likely to increase when the first several CCQs are
answered. The increase can happen when a surprising answer is
obtained, i.e. a yes answer is returned for low-probability
correspondence, or vice versa. This phenomenon indicates that, the
budget should be large enough to achieve satisfactory reduction of
uncertainty.

Another important finding is that, the uncertainty convergence to
zero in real implementation is much slower than that in the
simulation. A possible reason is that, we use a Bernoulli
distribution to model the error rate of workers. But in reality, the
error rate follows a much complex distribution, which may be related
to the dataset.

Lastly, we present the overall time cost of Single CCQ and Multiple
CCQ approaches in the real implementations, where totally 50 CCQs
are published and answered. As shown in Figure~\ref{fig:time}, the curves with larger $k$ tend to have less time cost. Please note that, the
case of Single CCQ is indicated with $k=1$. When k is increased, we get faster initial reduction on uncertainty, but the overall reduction tend to be limited. Actually, there are many
uncontrollable factors would affect the completion time, such as the
difficulty of the CCQs, the time of publication etc.\vspace{-3ex}
\nop{For instance,
it is observed that, the time cost for dataset ``ComplaintForm" is
higher than the other three. This may be because the number of attributes
in ``ComplaintForm" is higher, as shown in Table~\ref{datasets}, so the schemata are more complex for workers to understand.}
\begin{figure}[t]
    \includegraphics[width=8.5cm]{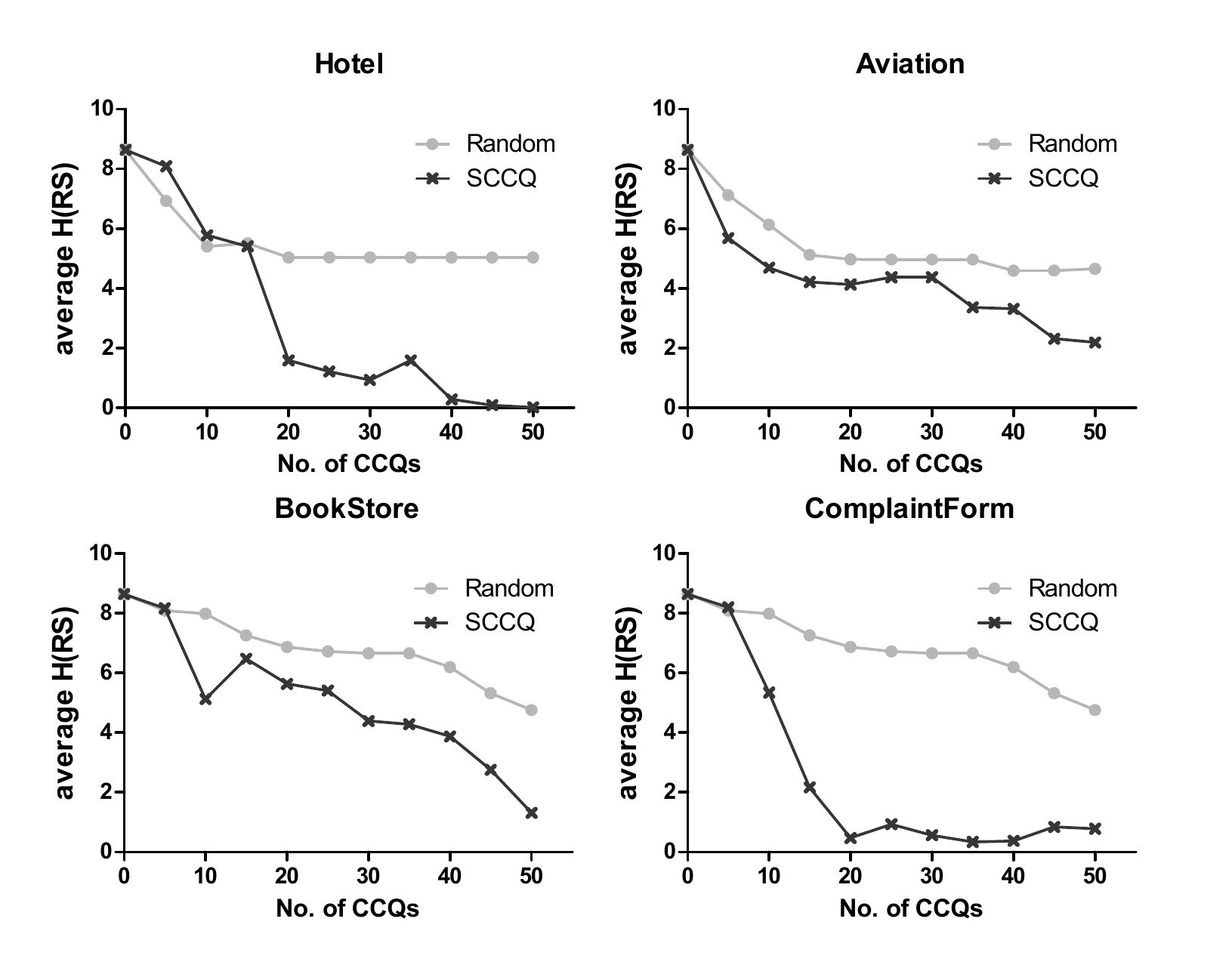}
      \vspace{-4.5ex}
    \caption{\footnotesize{Single CCQ v.s. Random - on Amazon Mechanical Turk}}
    \label{fig:SCCQReal}
   \vspace{-3ex}
\end{figure}
\begin{figure}[t]
    \includegraphics[width=8.5cm]{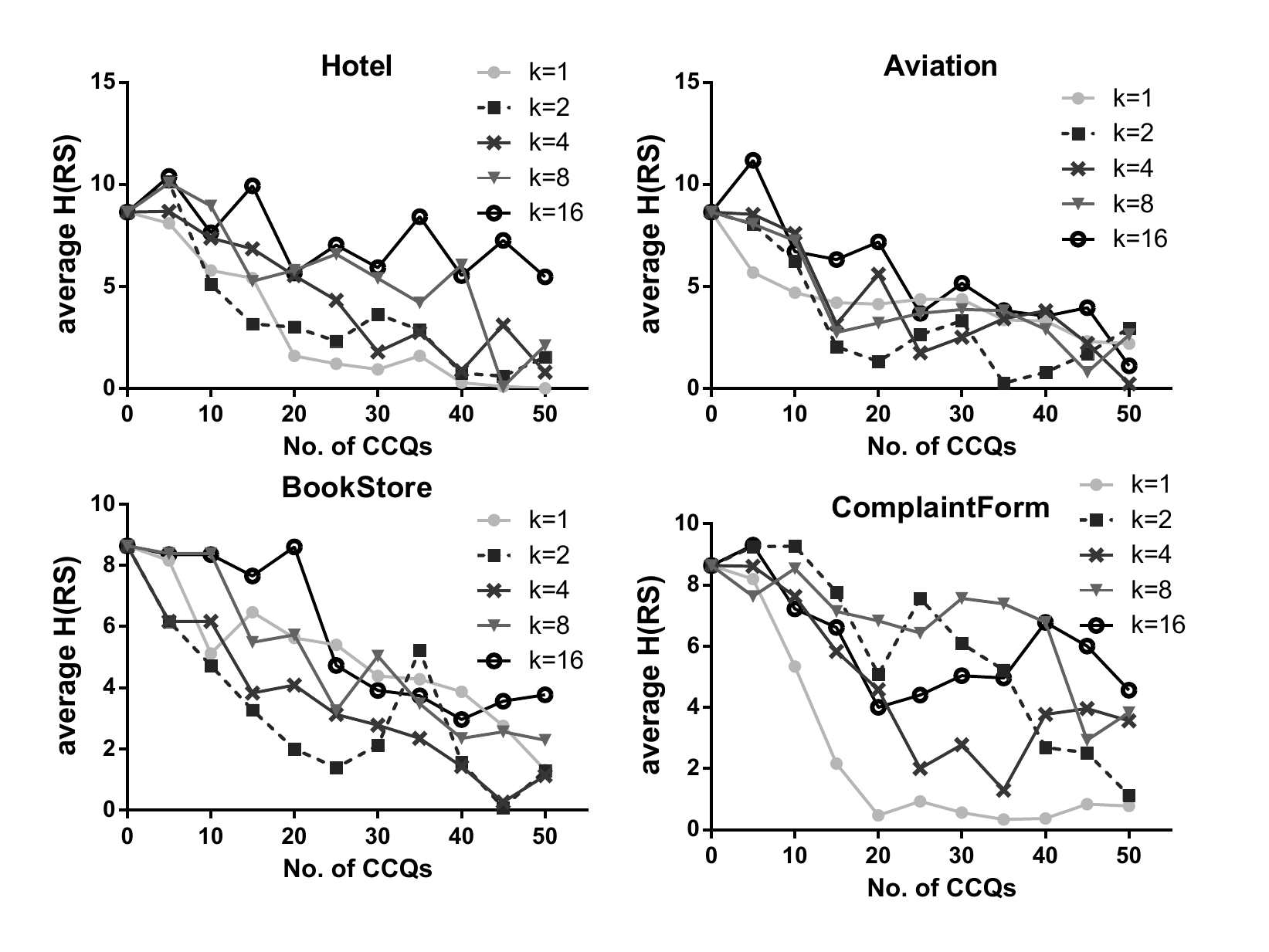}
   \vspace{-4.5ex}
    \caption{\footnotesize{Multiple CCQ with different k on Amazon Mechanical Turk}}
   \vspace{-4.2ex}
  \label{fig:MCCQReal}
\end{figure}
\begin{figure}[t]
    \includegraphics[width=8.5cm]{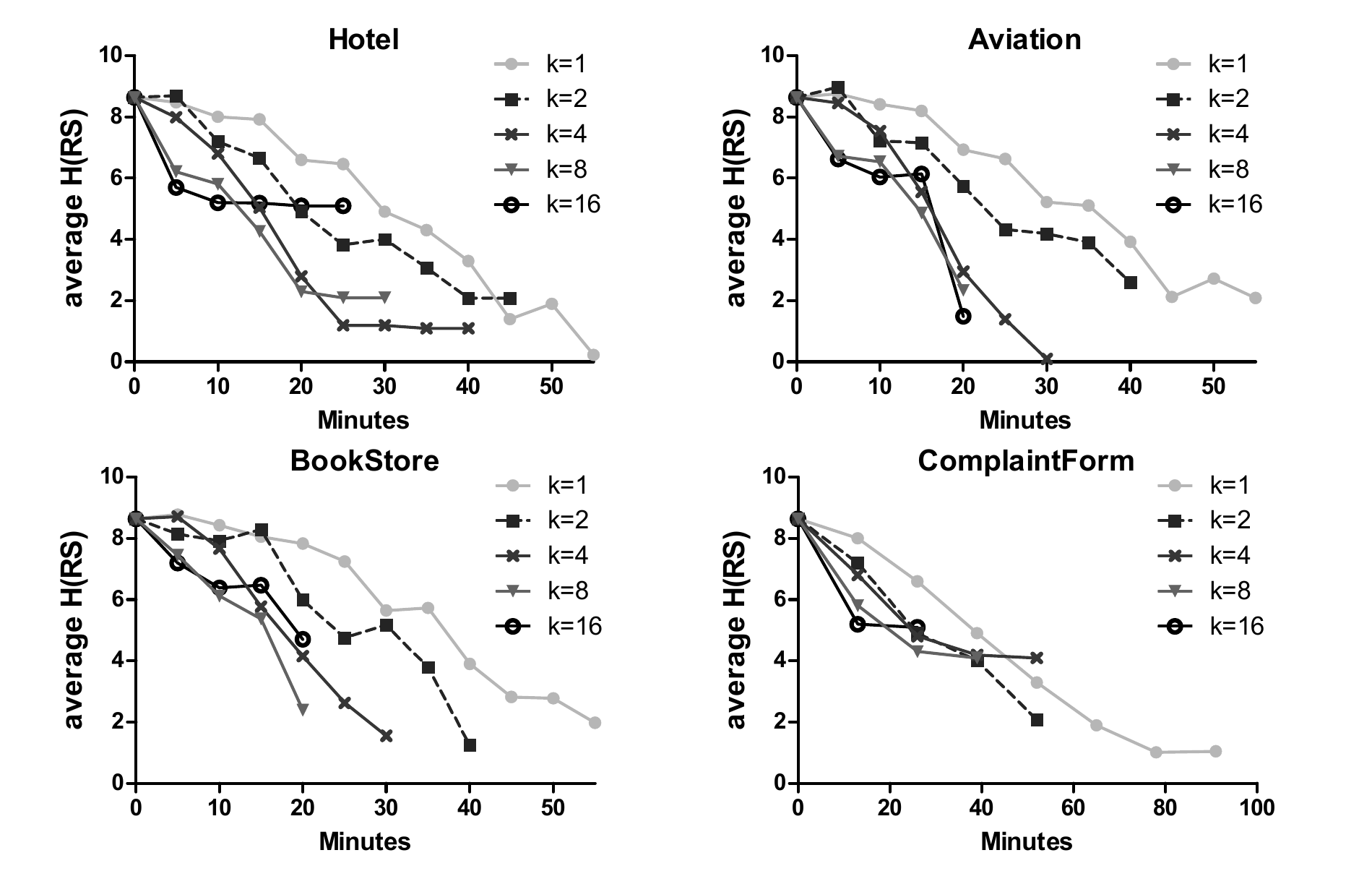}
      \vspace{-3.5ex}
    \caption{\footnotesize{Time Cost with different k on Amazon Mechanical Turk}}
   \vspace{-3.3ex}
    \label{fig:time}
\end{figure}
\subsection{Data Quality}
\vspace{-1ex}
In this subsection, we verify the correctness of our approaches, by
evaluating the precision and recall of the best matching, i.e. the
possible matching with the highest possibility after the uncertainty
reduction. Precision is computed as the ratio of correct
correspondences out of the total number of correspondences in the
correct matching (ground truth). Recall is computed as the ratio of
correct correspondences out of the total number of correspondences
in the correct matching. Since the performances are very similar on
different datasets, we merge the four datasets into one, and present
the precision and recall averaged from 40 runs.

Figure~\ref{fig:precision_recall} illustrates the quality of the best
matching after uncertainty reduction with budget $B=50$. The
suffixes ``\_S'' and ``\_R'' represent the data obtained from the
simulation and the real-world implementation on AMT, respectively. B mainly depend on how much money the HIT requester will pay for the task. In the simulation, the precision and recall are almost $100\%$. In
the real-world implementation, 50 questions by SCCQ make precision
and recall over $90\%$, which are significantly better than that of the ``machine-only'' methods when
k is small. However, in the real implementation, we find that when
$k$ is increased, the precision and recall tend to be decreased
dramatically. In particular, for cases $k=8$ and $k=16$, the MCCQ is
only slightly better than the \textit{Composition}. The reason is
twofold: first, comparing to SCCQ, there is averagely less
information for selecting CCQs in MCCQ; second, due to the
NP-hardness, we are only able to select CCQs that are near-optimal.

Recall that the motivation of MCCQ is to improve the time
efficiency. Therefore, we conducted another set of experiments where
time is the constraint, in order to investigate the relation between
$k$ and data quality. Explicitly, we preform SCCQ and MCCQ for $50$
minutes, without any limit on the budget. The precision and recall
are demonstrated in Fig~\ref{time_pr}. From the experimental
results, we conclude that \textbf{the MCCQ with large $k$ has outstanding
performance for time-constrained situations}.
Therefore, we conclude that \textbf{$k$
should be set to a small value when the budget is the main
constraint; whereas a large value is suggested for $k$ if
time-efficiency is the primary constraint}.\vspace{-3ex}
\begin{figure}[t]
    \hspace{-2ex}
    \includegraphics[width=8.5cm] {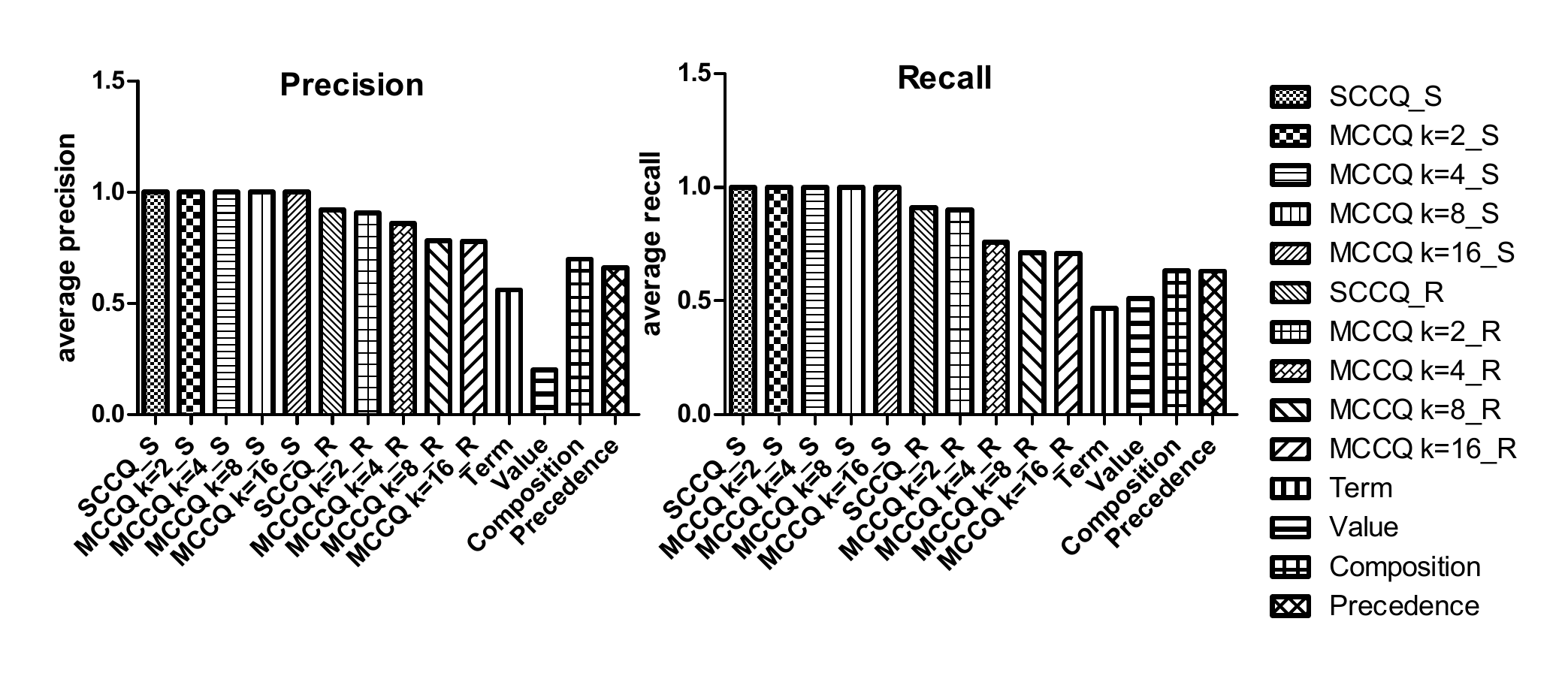}
\vspace{-4.5ex}

    \caption{\footnotesize{Data Quality with Budget Constraint- Precision \& Recall}}
\vspace{-4.5ex}
    \label{fig:precision_recall}
\end{figure}
 \begin{figure}[t]
    \includegraphics[width=8.5cm]{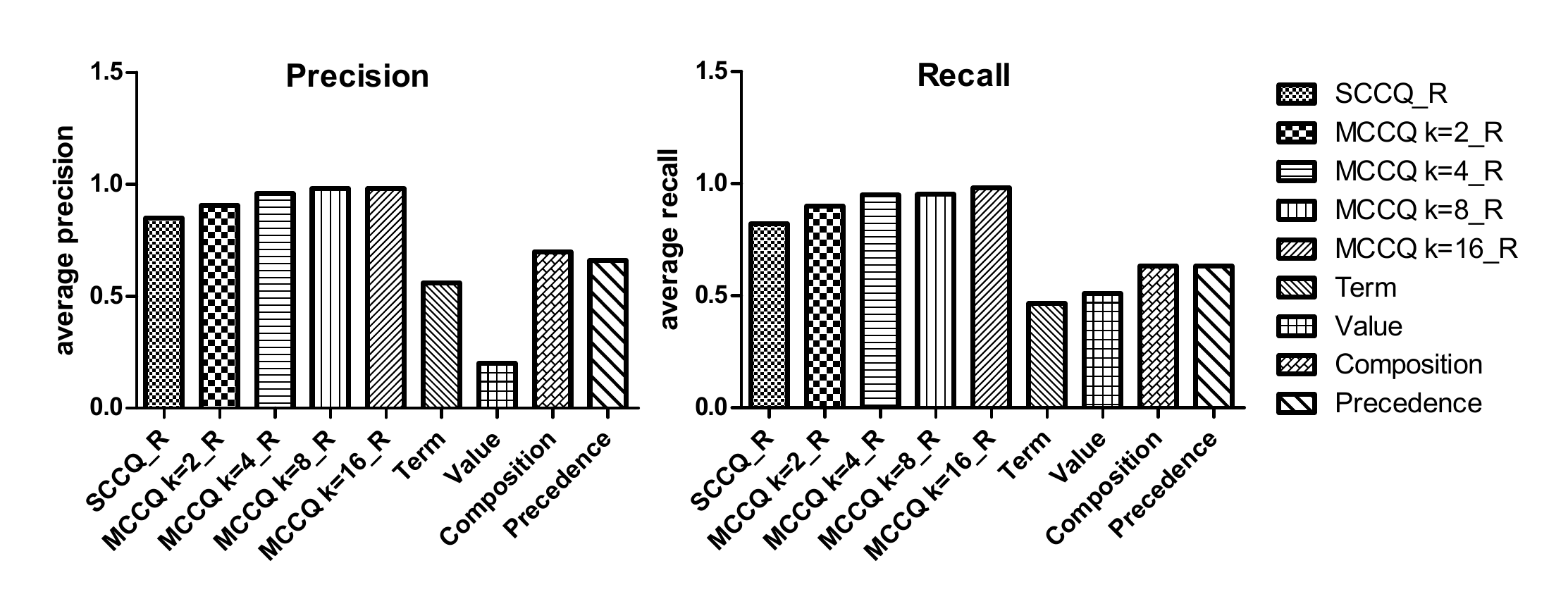}

  \vspace{-4.5ex}
    \caption{\footnotesize{Data Quality with Time Constraint - Precision \& Recall}}
    \vspace{-7ex}
    \label{fig:precision_recall_time}
     \label{time_pr}
\end{figure}

\subsection{New Experiments}
\vspace{-1ex}
With a more realistic model in this paper, we conduct experiments of MCCQ again. In simulation, firstly we randomly generate accuracy rates following uniform distribution on $[0.5,1]$ for all correspondences as their hardness attribute. We publish k initial CCQs with state ``waiting''. We still check the states of published CCQs every time unit.
For each time unit, each CCQ in state ``waiting'' may change to
``accepted'' with probability $P_0$ and each CCQ at state ``accepted'' may change to
``answered'' with probability $P_1$. Each answer is returned with an accuracy rate $P_{W_i}$ as the trustworthiness of the crowd. Accuracy rates $P_{W_i}$ also follows uniform distribution on $[0.5,1]$. We still set budget $B=50$ and    Figure~\ref{fig:MCCQ1} shows the performance of Multiple CCQ by varying k. Then in Figure~\ref{fig:MCCQReal2} we apply our MCCQ approach on Amazon Mechanical Turk. The difference between new experiments and the old ones in \cite{DBLP:pvldb/ZhangCJC13} is that we consider initial accuracy rates and different accuracy rates in each step. In \cite{DBLP:pvldb/ZhangCJC13}, assumption of theoretical results is that accuracy rates equal to 1, while in experiments we chose accuracy rates less than 1. Moreover, in this paper we obtain optimal upper bound and lower bound for entropy reduction, so that pruning rules are more efficient. These are major reasons that our new choices for CCQs are  comparatively better in terms of entropy reduction with less fluctuation.

At last we consider a new dataset with more attributes and we set $B=70$, $k=8$. Let $X$ be beta distribution $Beta(2,2)$. In Figure 12, we try different distributions for $P_{W_i}$. Line 1 shows $P_{W_i}$ follows uniform distribution on $[0.5,1]$ with mean 0.75 and variance $1/48$. In Line 2, $P_{W_i}=0.5X+0.5$, thus $P_{W_i}\in [0.5,1]$ with mean 0.75 and variance $1/80$. In Line 3, $P_{W_i}=0.4X+0.6$, thus $P_{W_i} \in [0.6, 1]$  with mean 0.8 and variance $1/125$. In line 4,  $P_{W_i}=0.6X+0.4$, thus $P_{W_i} \in [0.4, 1]$  with mean 0.7 and variance $9/500$. Line 5  shows the result in AMT. Comparing first four lines, we can see Line 3 perform best as $P_{W_i}$ has biggest mean and smallest variance. Line 4 perform worst since in practice we do not choose a crowd worse than 0.5.  \vspace{-2.5ex}
\begin{figure}[!t]
    \includegraphics[width=8.5cm]{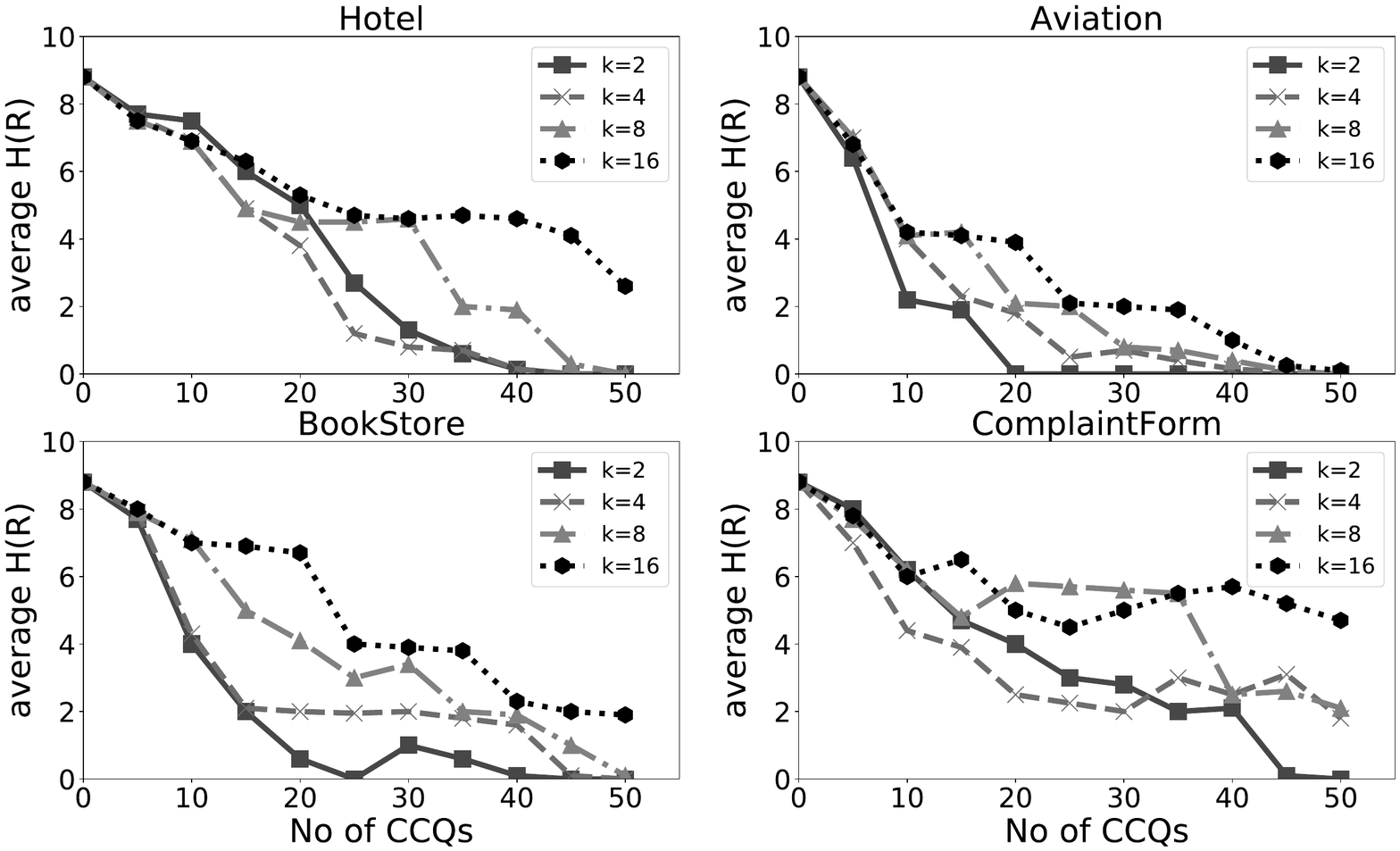}
       \vspace{-7ex}
       \caption{\footnotesize{Multiple CCQ with different k - Simulation(New)}}
       \vspace{-4ex}
    \label{fig:MCCQ1}
\end{figure}
\begin{figure}[t]
    \vspace{-6ex}
    \includegraphics[width=8.5cm]{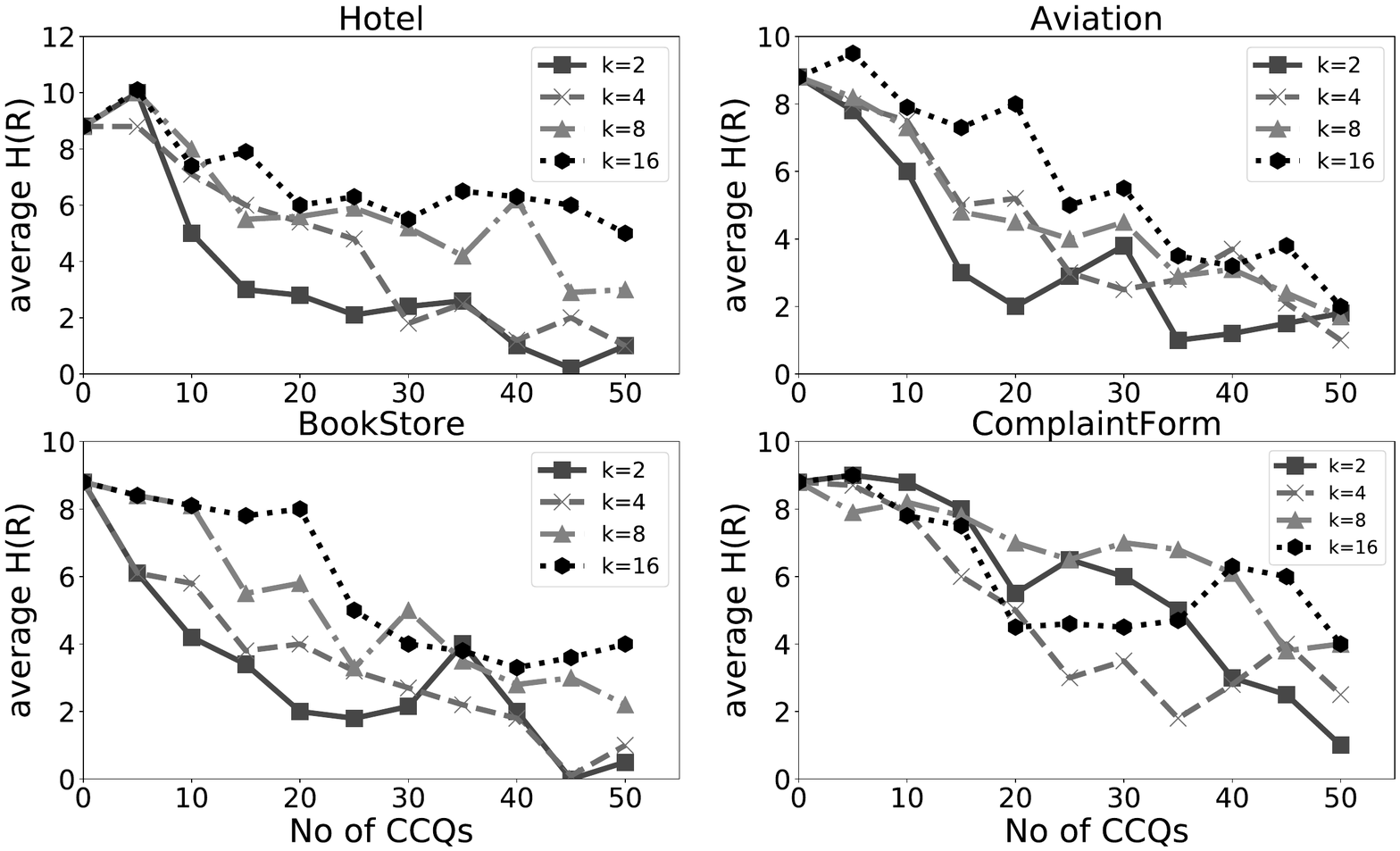}
    \vspace{-7.3ex}
    \caption{\footnotesize{Multiple CCQ with different k on Amazon Mechanical Turk(New)}}
    \vspace{-5ex}
  \label{fig:MCCQReal2}
\end{figure}
 \begin{figure}[t]
    \includegraphics[width=5cm]{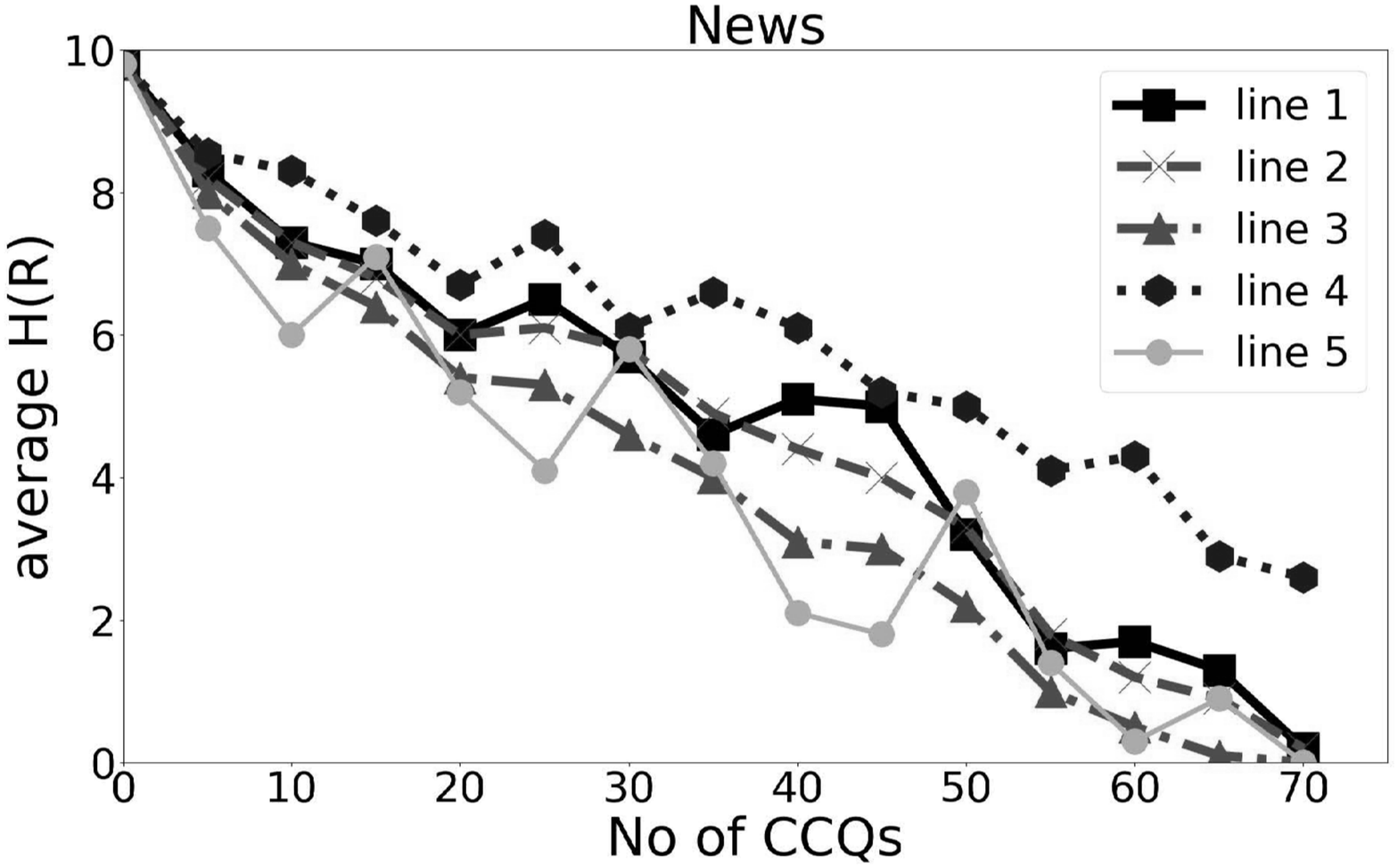}

  \vspace{-4.5ex}
    \caption{\footnotesize{MCCQ with different $P_{W_i}$}}
    \vspace{-4.5ex}
\end{figure}

%

\section{Related Work}
\vspace{-0.5ex}
\label{related}
\subsection{Uncertainty in Schema Matching}
\vspace{-0.5ex}
The model of possible matching, namely ``probabilistic schema mappings", was first introduced in
\cite{Dong:2007:DIU:1325851.1325930}. In their work, algorithmic
approaches generate a set of matchings between two schemata, with a probability attached to
each matching. After the collection of possible matchings is
determined, the probability of each correspondence can be computed
by summing up the probabilities of possible matchings in which the
correspondence is included.  Later, Sarma et
al. \cite{DBLP:conf/sigmod/SarmaDH08} used well-known schema
matching tools (COMA, AMC, CLIO, Rondo, etc.) to generate a set of
correspondences associated with confidence values between two
schemata. Then, the possible matchings are constructed from these
correspondences and data instances. A more
intuitive method of constructing possible matchings is proposed in \cite{DBLP:series/synthesis/2011Gal}. In detail, \cite{DBLP:series/synthesis/2011Gal}
generates top-k schema matchings by combining the matching results
generated by various matchers, and each of the k matchings is
associated with a global score. Then possible matchings are
constructed by normalizing the global scores. Additionally, the
model of possible matchings has been adopted in \cite{Gal:4812467}
as a core foundation for answering queries in a data integration
system with uncertainty.
Gal \cite{DBLP:journals/jods/Gal06} used the top-K schema mappings from a semi-automatic matchers to improve the quality of the top mapping. \cite{Dong:2007:DIU:1325851.1325930} \cite{Gal:4812467} and \cite{Qi:2007:FFE:1247480.1247499} were devoted to the parallel use of uncertain schema matchings, and proposed new semantics of queries.

The uncertainty in schema matching has been intensively studied, primarily focusing on the query processing in the presence of uncertainty. X.Dong et al. \cite{Dong:2007:DIU:1325851.1325930} concentrated on the semantics and properties of probabilistic schema mappings. We assume that a set of probabilistic schema matchings is provided by an existing algorithm, such as one of those mentioned above. How to efficiently process uncertain data is an orthogonal issue, which has been well addressed, such as \cite{DBLP:journals/pvldb/TongCCY12, DBLP:conf/icde/TongCD12,DBLP:conf/sigmod/HuangAKO09}.

A  probabilistic matching network model was established in \cite{nguyen2014pay} to reduce uncertainty of schema matching. Authors developed pay-as-you-go reconciliation approach. Probabilities of correspondences are defined in their model independently of schema matching tools. \cite{sagi2013schema} discussed schema matching prediction which is an assessment mechanism to support schema matchers in the absence of an exact match.
\vspace{-3ex}

\subsection{Crowdsourcing and Data Integration}
\vspace{-1ex}
Such as schema matching, some queries cannot be answered by machines only. The recent booming up of crowdsourcing brings us a new opportunity to engage human intelligence into the process of answering such queries (see \cite{DBLP:journals/cacm/DoanRH11} \cite{li2016crowdsourcedsurvey} \cite{chittilappilly2016survey} as survey for crowdsourcing). In general, \cite{DBLP:conf/sigmod/FranklinKKRX11} proposed a query processing system using microtask-based crowdsourcing to answer queries. Many classical queries are studied in the context of crowdsourced database, including max \cite{DBLP:conf/sigmod/GuoPG12}, filtering \cite{DBLP:conf/sigmod/ParameswaranGPPRW12}, sorting \cite{DBLP:pvldb/MarcusWKMM11} etc. In \cite{DBLP:conf/cidr/ParameswaranP11}, a declarative query model is proposed to cooperate with standard relational database operators. In \cite{ciceri2016crowdsourcing}, crowdsourcing is used for top-K query processing over uncertain data. As a typical application related to data integration, \cite{ERCrwodsourcing} utilized a hybrid human-machine approach on the problem of entity resolution. \cite{meng2017knowledge} studied knowledge base semantic integration using crowdsourcing.

\cite{mccann2008matching} engages crowdsourcing into schema matching. In particular, \cite{mccann2008matching} proposed to enlist the multitude of users in the community to help match the schemata in a Web 2.0 fashion. The difference between our work and \cite{mccann2008matching} is threefold: (1) From the conceptual level, ``crowd'' in \cite{mccann2008matching} refers to an on-line community (e.g. a social network group); while we explicitly consider the crowd as crowdsourcing platforms (e.g. Mechanical Turk). (2) The essential output of \cite{mccann2008matching} is determined by the ``system builders'', which means the end users still have to get involved in the process of schema matching. (3) We focus on the optimization between the cost (the number of CCQs) and performance (uncertainty reduction).\vspace{-2.5ex}
\subsection{Active Learning}
\vspace{-1ex}
Active learning is a form of supervised machine learning, in which a
learning algorithm is able to interact with the workers (or some
other information source) to obtain the desired outputs at new data
points. A widely used technical report is
\cite{DBLP:series/synthesis/2012Settles}. In particular,
\cite{DBLP:journals/corr/abs-1209-3686,DBLP:conf/aaai/ZhaoSS11}
proposed active learning methods specially designed for
crowd-sourced databases. Our work is essentially different from
active learning in two perspectives: (1) the role of workers in
active learning is to improve the learning algorithm (e.g. a
classifier); in this paper, the involvement of workers is to reduce the uncertainty of given matchings. (2) The uncertainty of answers are usually assumed
to be given before generating any questions; in this paper, the
uncertainty of answers has to be considered after the answers are
received, since we cannot anticipate which workers would answer our
questions. To our best knowledge, there is no algorithm in the field of active learning can be trivially applied to our problem.\vspace{-2.5ex}

\section{Conclusion and future work}
\vspace{-0.5ex}
\label{conclusion}
In this paper, we propose two novel approaches,
namely Single CCQ and Multiple CCQ, to apply crowdsourcing to reduce
the uncertainty of schema matching generated by semi-automatic
schema matching tools. These two approaches adaptively select and
publish the optimal set of questions based on new received answers.
Technically, we significantly reduce the complexity of CCQ selection
by proving that the expectation of uncertainty reduction caused by a
set of CCQs are mathematically equivalent to the join entropy of answers minus entropy of crowds. In addition, we obtain optimal bounds for uncertainty reduction, prove NP-hardness of
MCCQ Selection, and design an $(1+\epsilon)$ approximation
algorithm, based on its sub-modular nature. One challenge we overcome is to investigate difficulties of CCQs and trustworthiness of crowd-sourced answers by accuracy rates of crowds.

Uncertainty is inherited in many components in modern data
integration systems, such as entity resolution, schema matching,
truth discovery, name disambiguation etc. We believe that embracing
crowdsourcing as a component of a data integration system would be
extremely conductive for the reduction of uncertainty, hence
effectively improve the overall performance. Our work represents an
initial solution towards automating uncertainty reduction of schema
matching with crowdsourcing.

A future work regarding to MCCQ is that: in Theorem 4.1, we distribute $k$ CCQs to crowds each time. We obtain a formula of uncertain reduction under the assumption that we take back $k$ answers. In reality, we do not know how many CCQs can be
answered. We may withdraw or replace some CCQs after a waiting time.  The choice of next $k$ CCQs is best only when all $k$ CCQs are answered. Therefore  investigating a more realistic and complete model with answer rates(a difficult CCQ may has a probability that no one accept it)  may further help reducing the matching uncertainty.


%

\ifCLASSOPTIONcompsoc
  \section*{Acknowledgments}
\else
  \section*{Acknowledgment}
\fi

\ifCLASSOPTIONcaptionsoff
  \newpage
\fi



%
{
\begin{footnotesize}
\bibliographystyle{abbrv}
\bibliography{IEEEexample}
\end{footnotesize}
}
\footnotesize \textbf{Chen Jason Zhang} received the PhD degree from the Department of Computer Science and Engineering(CSE) at the Hong Kong University of Science and Technology(HKUST) in 2015. He is currently associate professor in Shandong University of Finance and Economics. His research interests include crowdsourcing and data integration.

\textbf{Lei Chen} received the PhD degree in Computer Science from the University of Waterloo, Canada, in 2005. He is currently a Professor in department of CSE, HKUST. His research interests include crowdsourcing, uncertain databases and data integration.

\textbf{H. V. Jagadish} is currently the Bernard A Galler Collegiate Professor of Electrical Engineering and Computer Science at the University of Michigan. He received his Ph.D. from Stanford University in 1985. His research interests include databases and Big Data.

\textbf{Mengchen Zhang} received his Ph.D. degree from department of Mathematics at HKUST in 2017. He is currently a research assistant in department of CSE, HKUST. His research interests include Stein's method in probability , crowdsourcing and data integration.

\textbf{Yongxin Tong} received the Ph.D. degree in department of CSE, HKUST in 2014. He is currently an associate professor in the School of Computer Science and Engineering, Beihang University. His research interests include crowdsourcing, uncertain data mining and social network analysis.

 \begin{figure}[t]
    \includegraphics[width=17cm]{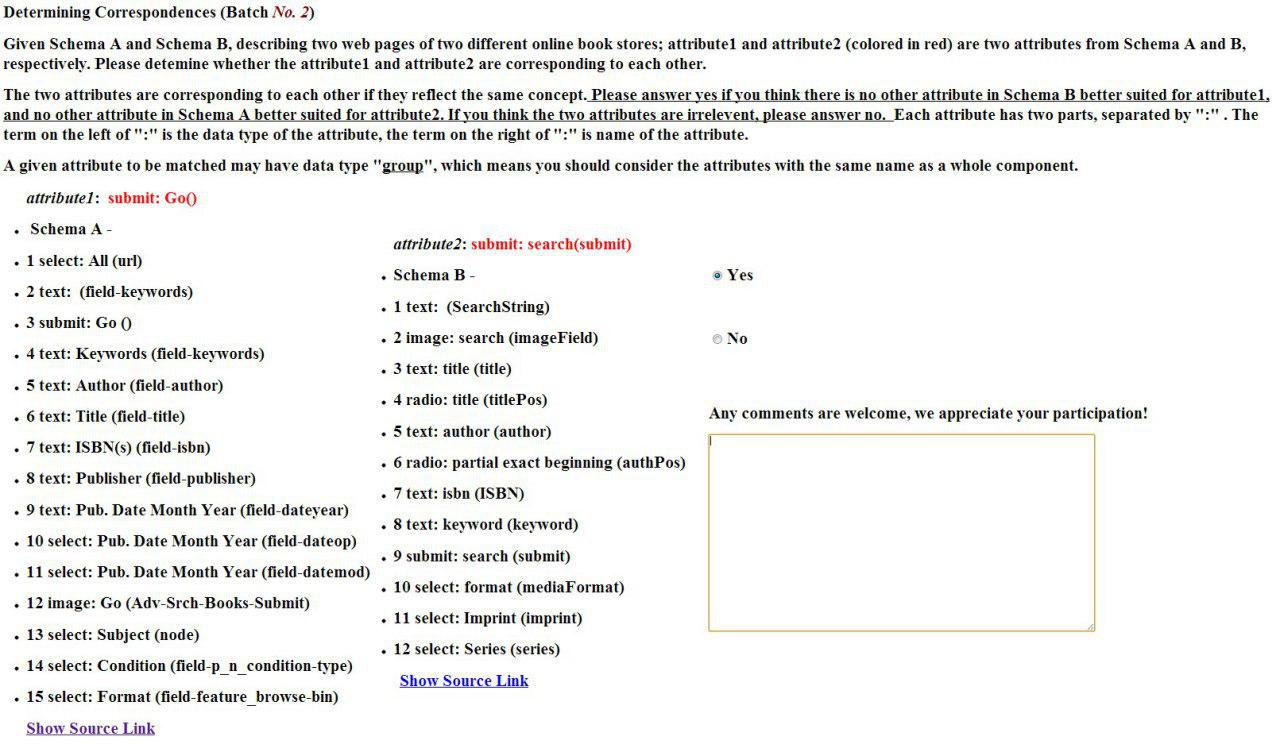}
    \caption{\footnotesize{screen shot of CCQ on AMT}}
\end{figure}

\end{document}